%% file: main.tex
\documentclass[11pt] {article}
\usepackage[utf8]{inputenc}
\usepackage{enumitem}
\usepackage{amssymb,amsmath, amsthm, amsfonts,amstext,mathrsfs}
\usepackage{bm}
\usepackage{listings}
\usepackage{xcolor}
\usepackage{grffile}
\usepackage{graphicx}
\usepackage{float}
\usepackage{setspace}
\usepackage[margin=1in]{geometry}
\usepackage{mathtools}
\usepackage{caption}
\usepackage{booktabs}
\usepackage[normalem]{ulem} 

\usepackage{color}
\definecolor{darkblue}{rgb}{0,0,0.5}
\definecolor{darkgreen}{rgb}{0,0.5,0}
\definecolor{darkred}{rgb}{0.9,0,0}
\usepackage[
    colorlinks=true,
    linkcolor=darkblue,
    urlcolor=darkblue,
    citecolor=darkgreen
]{hyperref}
\usepackage[misc]{ifsym}

\usepackage{indentfirst}
\usepackage{textcomp}
\usepackage{tikz}
\usepackage{tikz-cd}
\usepackage{bbm}
\usepackage{multirow}
\usepackage{mathpazo}
\usepackage{fancyhdr}
\usepackage{relsize}
\usepackage{mathpazo}
\usepackage{fancyhdr}
\usepackage{relsize}
\usepackage[all,cmtip]{xy}
\usetikzlibrary {arrows,automata,positioning }
\usetikzlibrary{positioning,calc}
\allowdisplaybreaks[1]
\usepackage{tikz,lipsum,lmodern}
\usepackage[most]{tcolorbox}

\usepackage{authblk}
\usepackage{mdframed}
\usepackage[strings]{underscore}

\usepackage{breakurl}
\usepackage{marginnote}

\usepackage{complexity}

\usepackage[capitalise]{cleveref}

\usepackage{titlesec}
\titlespacing{\paragraph}{0pt}{2ex}{0.15cm}

\usepackage{tikz}
\usetikzlibrary{arrows}

\tikzset{
    vertex/.style={circle,draw,minimum size=1.5em},
    edge/.style={->,> = latex'}
}

\usepackage{comment}
\excludecomment{non-ANON}
\includecomment{ANON}

\DeclareMathOperator{\negl}{negl}
\DeclarePairedDelimiter\bra{\langle}{\rvert}
\DeclarePairedDelimiter\ket{\lvert}{\rangle}
\DeclarePairedDelimiterX\braket[2]{\langle}{\rangle}{#1 \delimsize\vert #2}
\newtheorem{thm}{Theorem}[section]
\newtheorem{cor}[thm]{Corollary}
\newtheorem{lemma}[thm]{Lemma}
\newtheorem{prop}[thm]{Proposition}
\newtheorem{defn}[thm]{Definition}
\newtheorem{question}{Question}

\newcommand{\tinyspace}{\mspace{1mu}}
\renewcommand{\R}{\mathbb{R}}

\renewcommand{\C}{\mathbb{C}}

\newcommand{\N}{\mathbbm{N}}
\renewcommand{\H}{\mathcal{H}} 

\newcommand{\mcN}{\mathcal{N}}
\newcommand{\mcI}{\mathcal{I}}
\newcommand{\mcO}{\mathcal{O}}
\newcommand{\mcS}{\mathcal{S}}
\newcommand{\mcH}{\mathcal{H}}
\newcommand{\mcK}{\mathcal{K}}
\newcommand{\mcG}{\mathcal{G}}
\newcommand{\msB}{\mathscr{B}}
\newcommand{\msU}{\mathscr{U}}
\newcommand{\epr}{\Phi_{\text{\tiny EPR}}}
\newcommand{\eprn}{\epr^{\otimes n}}
\newcommand{\wtd}{\widetilde}
\newcommand{\Tr}{\mathrm{Tr}}
\newcommand{\Irr}{\mathrm{Irr}}

\newcommand{\I}{\mathcal{I}} 
\newcommand{\ang}[1]{\langle #1 \rangle}
\newcommand{\norm}[1]{\lVert\tinyspace #1 \tinyspace\rVert} 
\newcommand{\abs}[1]{\lvert\tinyspace #1 \tinyspace\rvert} 

\newcommand{\Id}{\mathbbm{1}}

\newcommand{\nolines}[1]{\multicolumn{1}{c}{#1}}

\renewcommand{\hat}[1]{\widehat{#1}} %
\let\epsilon=\varepsilon %

\newcommand{\arr}{\rightarrow}

\newcommand{\SV}{\class{OTS}} %
\newcommand{\OTS}{\class{OTS}}

\title{Quantum Delegation\\ with an Off-the-shelf Device
}

\author[1,2] {Anne Broadbent}
\author[1,2]{Arthur Mehta}
\author[3,4]{Yuming Zhao}
\affil[1]{Department of Mathematics and Statistics, University of Ottawa}
\affil[2]{Nexus for Quantum Technologies, University of Ottawa}
\affil[3]{Institute for Quantum Computing, University of Waterloo}
\affil[4]{Department of Pure Mathematics, University of Waterloo}
\date{}

\begin{document}

\pagenumbering{roman}
\setcounter{page}{1}

\maketitle

\begin{abstract}
Given that reliable cloud quantum computers are becoming closer to reality,  the concept of delegation of quantum computations and its verifiability is of central interest. Many models have been proposed, each with specific strengths and weaknesses. Here, we put forth a new model where the client trusts only its classical processing, makes no computational assumptions, and interacts with a quantum server in a \emph{single} round. In addition, during a set-up phase, the client specifies the size $n$ of the computation and receives an untrusted, \emph{off-the-shelf (OTS)} quantum device that is used to report the outcome of a single measurement.

 We show how to delegate polynomial-time quantum computations in the OTS model. This also yields an interactive proof system for all of $\QMA$, which, furthermore, we show can be accomplished in statistical zero-knowledge. This provides the first relativistic (one-round), two-prover zero-knowledge proof system for~$\QMA$.

As a proof approach, we provide a new self-test for $n$ EPR pairs using only constant-sized Pauli measurements, and show how it provides a new avenue for the use of simulatable codes for local Hamiltonian verification. Along the way, we also provide an enhanced version of a well-known stability result due to Gowers and Hatami and show how it completes a common argument used in self-testing.

\end{abstract}
\newpage
\setcounter{tocdepth}{2}
\tableofcontents
\newpage

\section{Introduction}\label{sec:intro}
\input{Introduction.tex}

\section{Preliminaries}\label{sec:prelims}
\input{Preliminaries.tex}

\section{Low-weight Pauli braiding test}\label{Sec:LWPBT}
\input{LWPBT.tex}

\section{Rigidity of low-weight Pauli braiding test}\label{sec:rigidity}
\input{Rigidity.tex}

\section{Modified Hamiltonian game}\label{sec:Game}
\input{HamiltonianGame.tex}

\section{Zero-knowledge proof system}\label{Sec:ZK}
\input{ZK.tex}

\section{Off-the-shelf model}\label{sec:SkepticalVerifierModel}
\input{SkepticalVerifierModel.tex}

\input{main.bbl}

\end{document}

%% file: Introduction.tex
\pagenumbering{arabic}
\setcounter{page}{1}

In an interactive proof system, a computationally-bounded verifier interacts with a powerful prover in order to verify the truthfulness of an agreed-upon problem instance. Starting with $\QMA$,
and followed by  $\QIP$ and $\QMIP$ (among others), \emph{quantum} interactive proof system, (in which  the verifier  is  \emph{quantum} polynomial-time)  were defined and studied~\cite{Wat00,Wat03,KM03}.

Yet, these quantizations depend crucially on the tacit assumption that the verifier has access to \emph{trusted} quantum polynomial-time verification.
Given the current state-of-the-art in quantum computation development, the inherent difficulty at characterizing quantum systems, and the fact that there is no way to reliably verify the trace of a quantum computation, there is ample evidence that this assumption may be questionable. Indeed, despite impressive technological improvements, we may ultimately have to contend with a reality where quantum computers are never as trustworthy or reliable as classical devices. This prospect has motivated consideration of models where the verifier has access to very limited but trusted quantum functionality \cite{ABEM17arxiv,Bro18,FK17}, or where the verifier is entirely classical and the prover is computationally bounded \cite{Mah18}, while another class called $\MIP^*$ models an efficient classical verifier interacting with several isolated, unbounded quantum provers~\cite{CHTW04}. Each approach provides advantages and encounters challenges: early quantum servers will be expensive and thus all else equal, requiring a single prover is preferable; on the other hand, existing single-prover protocols either require a trusted device or make computational assumptions.
Multi-prover protocols utilize powerful device-independence techniques which avoid these assumptions but at the high cost of requiring several powerful provers and requiring isolation.\looseness=-1

The current zeitgeist in this field allows for imaginative considerations of how we describe and model tasks in a quantum world.  These approaches have in common that instead of considering the straightforward quantum analog of classical protocols, we strive to make considerations that are naturally motivated in the quantum setting\footnote{See, for instance, the recent work on the complexity of preparing quantum states and unitaries~\cite{RY22}.}. Here, we continue on this momentum and introduce a novel approach to proof verification, where the set-up itself can only be motivated in the quantum setting. To this end, we consider the following question:

\begin{question}
\label{question:skeptical}
 \emph{What is the expressive power of the class of \textbf{relativistic}, interactive proof systems with a single quantum prover,  and a classical verifier having access to an \textbf{off-the-shelf} \textbf{untrusted quantum device}?}
\end{question}

\paragraph{Off-the-shelf Device.} We call the above model the \emph{off-the-shelf (OTS)} model since it models the fact that the verifier, in addition to interacting with a standard prover, has access to a device that is (1) generic (it does not depend on the instance of the problem to be solved, only on the instance size), (2) efficient (for completeness, polynomial resources are suffice) (3) completely untrusted (for soundness, there are no assumptions on its computational power or inner-workings). Importantly, \emph{relativistic} refers to a 1-round protocol; this is desirable for its relative ease in enforcing isolation\footnote{A relativistic protocol is highly desirable in the multi-prover scenario since isolation can be enforced using relative position and response times~\cite{CL17,Gri19}.}.

Operationally, we imagine the OTS model as the prover providing the verifier with such a generic, off-the-shelf device ahead of the proof verification. In particular, the preparation of such a device in terms of its capabilities is independent of the particular problem instance, although we do allow dependence on its size. Once in possession of this device, the verifier may query the prover and simultaneously use a single measurement from the off-the-shelf device, which leads the verifier to \emph{accept} or \emph{reject}. The figures of merit for the interactive proof system are the usual \emph{completeness} and \emph{soundness}.

\begin{figure}[!htbp]
  \centering
  \includegraphics[width=0.65\textwidth]{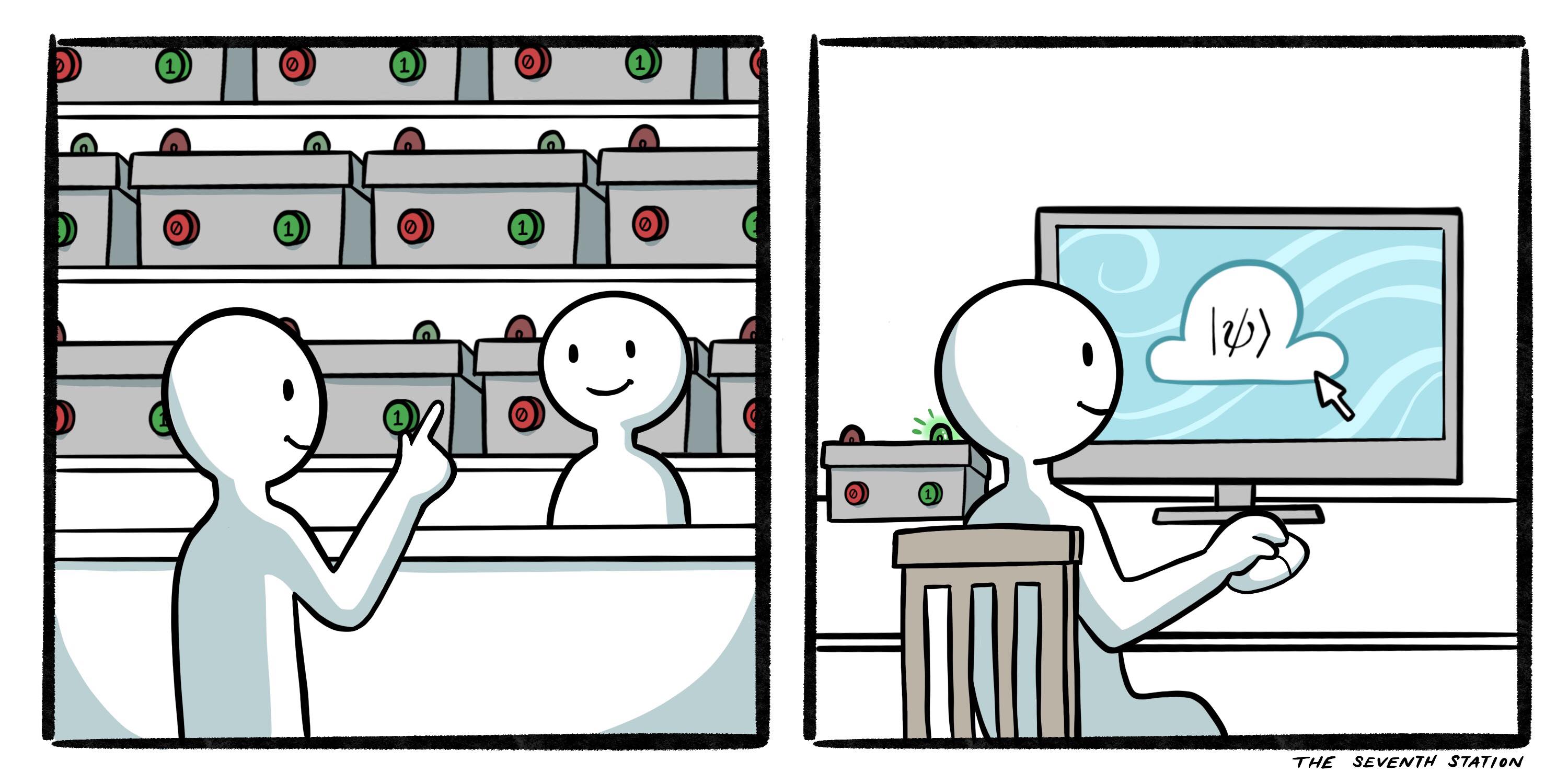}
  \caption{During the set-up the verifier selects an off-the-shelf device based on the required size of the problem instance. Afterward, the verifier is free to select any language and instance and interacts in a single round with both the  prover and the off-the-shelf device, leading to the \emph{accept/reject} output of the verifier.}
  \label{fig:Delegated quantum computation}
\end{figure}

Since the OTS scenario models aspects of near-term proof verification using untrusted quantum devices, we naturally wish to understand how it relates to some of the most relevant and studied properties of interactive proof systems:
\begin{question}
\label{question:zk}
\emph{Can the OTS model provide  novel approaches to \textbf{zero-knowledge proof systems} and to \textbf{delegated quantum computation}?}
\end{question}

\paragraph{Zero-Knowledge Proof Systems.}
\emph{Zero-knowledge (ZK)} proof systems capture the counter-intuitive notion that an interactive proof can be simultaneously convincing, while also completely concealing the inner-workings of the prover; such proof systems play an influential role in many areas of complexity and cryptography as well as in cybersecurity since they prevent reverse-engineering. We are interested here in \emph{statistical} ZK, meaning that the concealing property holds against an unbounded verifier. Given our new OTS model, one of the first questions we thus ask is whether proof systems in the OTS model can be made~ZK.


\paragraph{Delegated Quantum Computation.}
In \emph{delegated quantum computation}, a computationally-weak client outsources a quantum computation to a more powerful (but computationally-bounded) device in a way that the result is verifiable.  Starting with the OTS model, we can scale it down (such that the prover is no longer unbounded, but merely quantum polynomial-time; the verifier and off-the-shelf device remain unchanged).
The question then becomes: via a relativistic interaction, can a classical polynomial-time client, together with an OTS device, verifiably outsource a quantum computation of their choosing to an isolated quantum polynomial-time device, assuming that the size of the quantum computation is compatible with the parameters of the OTS device?

\paragraph{Summary of Results.}~In this work, we make important steps towards answering the above questions:\looseness=-1
\begin{itemize}
\item  We show that any language in $\QMA$ has a statistical ZK proof system in the OTS model.

\item  We show that the above OTS proof system can be adapted for delegated quantum computation for any problem in $\BQP$, while remaining ZK and in the OTS model.

\end{itemize}

\subsection{Context}
We now give an overview of background material, together with a summary of the current state-of-the-art approaches to interactive proof systems and delegated computation in the quantum setting, focusing on elements that are relevant to our main questions, and highlighting where existing approaches fail in their applicability to our scenario.

\subsubsection{Classical and quantum interactive proof systems} In the model of interactive proof systems ($\IP$), an efficient classical verifier interacts with an all-powerful and untrusted prover in order to verify the correctness of a statement~\cite{GMR89}. We note that class $\NP$ corresponds to a single-message interaction (with $\MA$ being in probabilistic version), while $\AM$ incorporates a single round (\emph{i.e.}, \emph{two} messages).

In \emph{a multiprover interactive proof system} ($\MIP$), a verifier interacts with multiple isolated provers~\cite{BGKW88}. Each of the models above has been \emph{quantized}, \emph{i.e.,} extended to the setting where some (or all) of the parties are quantum. This is captured, \emph{e.g.} by the classes $\QMA$ (the quantum version of $\MA$), $\QIP$ (the quantum version of $\IP$) and ${\MIP}^*$ (a version of a multi-prover interactive proof system ($\MIP$) where the unbounded provers share entanglement). Groundbreaking results have characterized some these quantum classes, \emph{e.g.} $\QIP = \PSPACE$ \cite{JJUW11} and ${\MIP}^*= \RE$~\cite{JNV+20arxiv}.

\subsubsection{Zero-knowledge}
\label{sec:intro-ZK}
A strong motivation for the study of interactive proof systems is the connections to the counter-intuitive concept of a zero-knowledge proof system \cite{GMW91,BOGG88}. Informally, a proof system is \emph{zero-knowledge} when the verifier is unable to learn anything beyond the fact that the agreed-upon instance is true. This is more formally treated by establishing the existence of a \emph{simulator} which can reproduce the transcript of the interaction.

Zero-knowledge proof systems were first extended to the quantum setting by Watrous \cite{Wat09}, who considered the setting where the verifier has access to a trusted polynomial-time quantum device. Subsequently, it was shown that under certain cryptographic assumptions, all problems in $\QMA$ admit a zero-knowledge proof system \cite{BJSW16,BJSW20,BG22} (again, assuming the verifier has trusted polynomial-time quantum computation). There have been several approaches in the case of a fully classical verifier. Vidick and Zhang showed that argument protocols can be made to satisfy the zero-knowledge property \cite{VZ20}. Recent work by Crépeau and Stuart~\cite{CS23arxiv} provides a two-prover one-round zero-knowledge proof system for $\NP$. The work of Chiesa, Forbes, Gur and Spooner provides a two-prover zero-knowledge
proof system for $\NEXP$~\cite{CFGS18}, however, their work requires polynomially many rounds of interaction. Work due to Grilo, Yuen, and Slofstra~\cite{GSY19} shows that any proof system for ${\MIP}^*$ can be made zero-knowledge at the cost of adding four additional provers. Although these works provide inspiration for studying zero-knowledge proof systems in the OTS model, as far as we are aware, they do not directly contribute to our main question on ZK.  In fact, according to the current state-of-the-art, an implicit open question~\cite{Gri19} is the following: \emph{``Does there exists a \emph{relativistic} zero-knowledge proof system for $\QMA$ with two provers and a classical verifier?''}. We emphasize that our OTS model takes this question further, by requiring one of the provers to operate generically and independently of the problem instance.


\subsubsection{Delegation of quantum computations}
\label{sec:intro-delegated} Delegated quantum computation allows a computationally-weak classical client to delegate a computational task to an untrusted, polynomial-time quantum server. Under certain conditions, an interactive proof system leads in a straightforward way to a protocol for delegated quantum computation. Typically, this is achieved if the interactive proof system captures \emph{e.g.}~$\QMA$, and furthermore, given the witness, the prover is efficient; it is also relevant that the $\QMA$ witness is used in such a way that we can scale down the proof system in order to achieve a delegation protocol for~$\BQP$ (\emph{e.g.}~\cite{Gri19})\footnote{$\BQP$ is closed under complementation, hence this is sufficient for delegation}. The sketch above is also applicable to the scenario of multiple servers. Note that because of the resemblance between the models of the interactive proof system and delegated quantum computation, we occasionally confound the two --- using the complexity class acronym to refer to the interaction pattern between prover(s) and verifier --- but we emphasize that in delegated quantum computation \emph{protocols}, the server is always computationally bounded (as opposed to a prover in interactive \emph{proof systems}).

Following Reichardt, Unger, and Vazirani~\cite{RUV13}, who showed a delegated quantum computation for the setting of ${\MIP}^*$, much progress was made, aiming at improving parameters and techniques; despite these efforts,  as far as we are aware, none of the existing works are applicable to our model.
Notable here is the work of~\cite{CGJV19} which uses \emph{quasilinear} resources for \emph{both} servers, and achieves at best a constant round complexity, as well as~\cite{Gri19} which is the first 2-server, 1-round (relativistic) protocol for delegated quantum computations, but uses the full polynomial-power of both servers.

We note that in the protocol of~\cite{RUV13}, as well as for most of the related works, we have the property that the protocol for delegated quantum computation can be \emph{scaled up} to an interactive proof system for $\QMA$ in the following way. We let the unbounded prover derive a $\QMA$ witness, then apply the delegation protocol for the circuit of the $\QMA$ verifier, where the \emph{other} prover does the computation on the witness, as teleported by the initial prover (and where the outcome of the teleportation measurement is reported to the verifier). In this way, we can closely relate interactive proof systems for $\QMA$ with two-prover protocols for delegated quantum computations.

\subsection{Technical context}

We now introduce two recent and highly successful techniques in the areas of quantum multi-prover interactive proof systems, delegated quantum computation and zero-knowledge proof systems:  \emph{self-testing} and \emph{quantum simulatable codes}. These techniques are pivotal for our work.

\subsubsection{Self-testing}
\label{sec:intro-self-testing}

\emph{Self-testing} (also called \emph{device-independence}) is a ubiquitous and powerful technique in the study of ${\MIP}^*$ and related delegation protocols. The concept was introduced by Mayers and Yao~\cite{MY04}. Informally, a protocol \emph{self-tests} a particular state or measurement when this state/measurements (or an equivalent version thereof) are required for obtaining the maximal acceptance probability. The most well-known examples are the non-local games known as the CHSH game and the Magic Square game \cite{CHSH69,Mer90,Per90,Tsi93}. Subsequently, numerous works have enriched our understanding of self-testing and its applications to delegated quantum computation, \emph{e.g.},  \cite{MYS12,McK17,Col17,CRSV18,CGJV19,NV17,NV18}.

Many of the above works arrived somewhat before the complementary mathematical formalism and careful examination of the basic properties of self-testing. Current approaches to formalize self-testing use the theory of approximate representation theory of groups and $C^*$-algebras \cite{slo19,Slo16arxiv,MPS21arXiv}, and a variety of very fundamental questions regarding self-testing have only recently been asked and examined \cite{CMMN20,MNP21arxiv,MS23}. These formalisms, and especially their operationally-useful \emph{approximate} versions utilize a key stability result due to Gowers and Hatami which allows one to relate approximate representations to exact representations~\cite{GH17}.

\subsubsection{Simulatable codes}
\label{sec:intro-sim-codes}
Recent works by Grilo, Yuen, and Slofstra \cite{GSY19}, as well as Broadbent and Grilo \cite{BG22} introduce the notion of simulatable codes as a tool for establishing zero-knowledge proof systems and protocols in the quantum setting. The idea is to use techniques from quantum error-correcting codes to create a ``simulatable'' witness or proof for use in the verification process. Here the witness is \emph{simulatable} in the sense that there is an efficient classical algorithm
 which can reproduce the description of the local density matrix of the witness on any small enough subspace. This is a pivotal tool in establishing zero-knowledge, and the application of the technique consists in developing  a verification protocol, (or verification circuit in the case of~\cite{BG22}) which verifies such simulatable witnesses; this can  then be applied to the situation of encoding \emph{e.g.}, a witness for~$\QMA$ into a simulatable code~\cite{BG22}.

We note that the full power of simulatable codes is best appreciated when thinking of them as for encoding a quantum state: thanks to such codes, we achieve a situation where local descriptions are easy to compute, while \emph{global} correlations might be hard --- a situation that has no classical equivalent since a series of local descriptions trivially define a global description.

\subsection{Contributions}

We now give more details and motivation for our model and an overview of our main contributions at the conceptual level.

\subsubsection{Model}

As introduced earlier, we are interested in modeling near-term proof verification and delegation of quantum computations. To this end, we propose a new paradigm that is particularly relevant to the quantum scenario: a verifier having access to an OTS device. To motivate the model, consider that the complexity class~$\QMA$ models a verifier having access to fully-trusted polynomial-time quantum computation. While such a verifier is \emph{skeptical} of the prover (and thus needs to verify the claimed proof independently), in the quantum case, a new level of skepticism is possible, namely that the verifier's quantum processing is untrusted. A common solution in this case is to postulate \emph{two} (or more) untrusted and all-powerful devices together with a classical verifier; this is the realm of~${\MIP}^*$.  In this work, we propose a new paradigm that treats the provers asymmetrically. Starting with a conventional two-prover interactive proof system, we ask that only \emph{one} of the provers do the heavy lifting (via its unbounded computational capabilities), with the second prover becoming efficient and completely generic (for completeness, this prover need not even be given a description of the task at hand; soundness, however, is shown against \emph{two} unbounded provers).


We denote $\OTS$ the set of all languages~$L$ that can be decided under a constant completeness-soundness gap, in the model that follows.  Before the instance $x \in L$ is selected, the classical verifier is provided with an untrusted off-the-shelf device which only depends on a parameter $n$, indicating the size of the problem instance (without loss of generality, we can assume that the prover provides such OTS device). For completeness, such a device shares an entangled state $\ket{\psi}_n$ with a quantum prover and will be purported to perform efficient measurements from a predetermined list of available options. The verifier may select any choice $x \in L$ provided $|x| \leq n$ and \emph{simultaneously} uses a single question to the prover and to the device; the verifier then determines whether or not to accept based on the responses. We stress that OTS proof systems are sound against both an unbounded prover and unbounded~OTS.

We observe that  $\OTS$ is a refinement of and thus contained in ${\MIP}^*$. In the other direction, the OTS model is a generalization of $\AM$, where the otherwise classical verifier has additional 1-round query access to a small, off-the-shelf quantum device.  In summary, we have the following straightforward containments (see also \Cref{fig:protocol})
\begin{equation}
 \AM \subseteq \OTS \subseteq {\MIP}^*
\end{equation}

In a classical proof system, an $\OTS$ can be understood as an instance-independent hardware token. This device can be used to provide a commitment for a zero-knowledge proof system for $\NP$~\cite{GMW91}; what is more, the one-time property of the OTS can be used as an oblivious transfer device, which then yields a non-interactive zero-knowledge proof system for~$\NP$~\cite{Kil88}.  We note that in the quantum case, our model requires a fully classical verifier and hence the case of zero-knowledge for $\QMA$~\cite{BJSW16,BG22} in the OTS model is much more complex, and a classical-verifier analogue to the $\NP$ proof systems above is not directly applicable. Other approaches based on using the OTS as a one-time memory~\cite{BGS13} also run into a roadblock due to the fact that we require a fully classical verifier.

\begin{figure}[!htbp]
    \centering
    \includegraphics[scale=0.12]{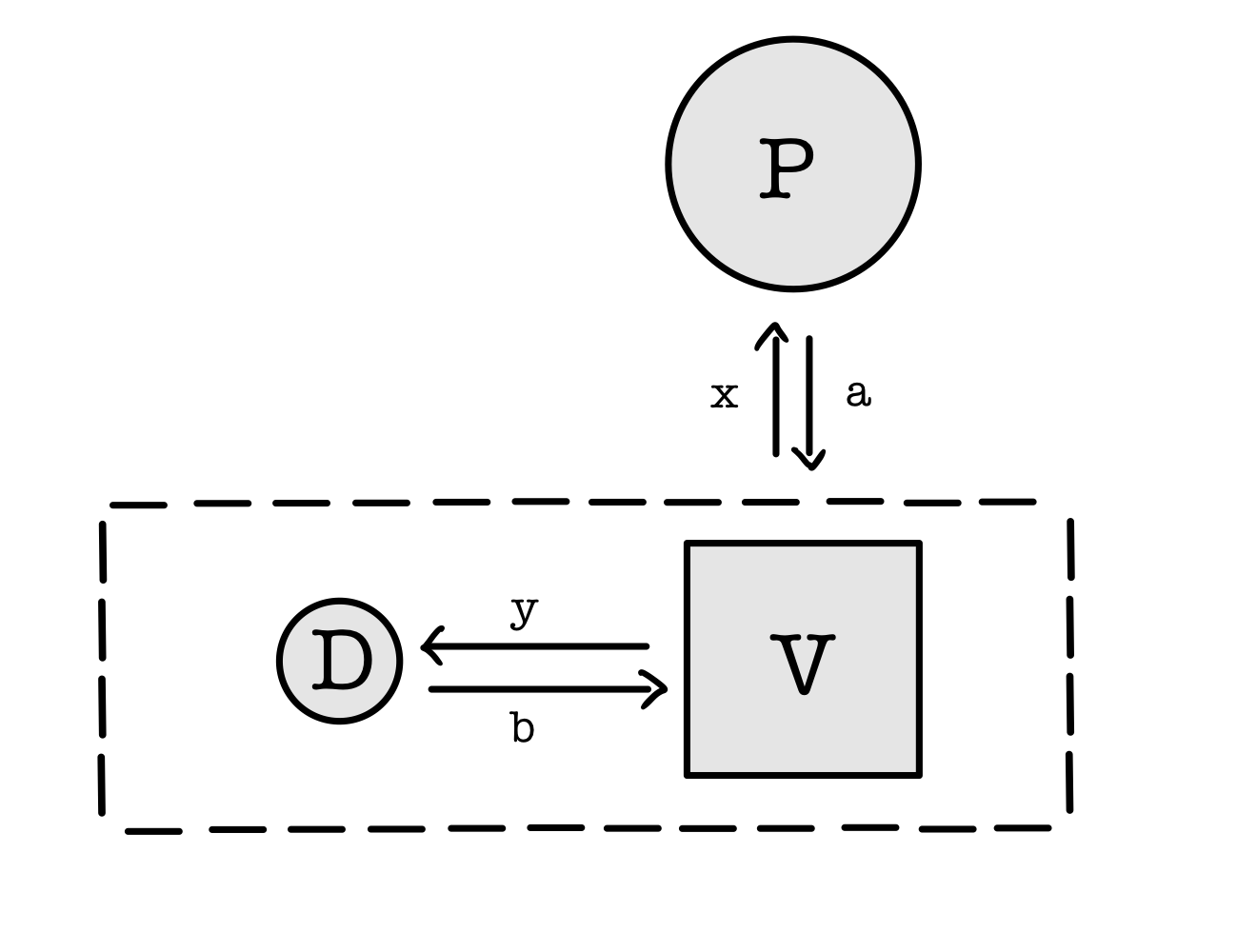}
    \caption{Off-the-shelf (OTS) proof system. $P$ is the quantum prover, $V$ is the classical verifier, and $D$ a  rudimentary off-the-shelf-device; each  arrow represents a single classical message. We can interchangeably think of the model as a strengthening of the verifier in~$\AM$ ($V$ having access to an additional OTS device~$D$), or a weakening of ${\MIP}^*$ (one of the provers, $D$, is severely restricted).}
    \label{fig:protocol}
\end{figure}

\subsubsection{OTS proof systems for $\QMA$} Our first result is that any language in $\QMA$ is also in $\OTS$. An interpretation of this result is that starting with a conventional proof system for $\QMA$, we can exchange the unwavering trust of the verifier in its quantum verification process for a classical verifier with two new features: (1) the verifier has access to an untrusted, and instance-independent, off-the-shelf quantum device; and (2) the verifier interacts with the prover (and the device) in a single simultaneous round. In summary, we thus have:

\begin{thm}\label{thm:InformalQMA} (Restated as part of \Cref{thm:QMA in Skep})
$ \QMA \subseteq \SV\,.$
\end{thm}

\subsubsection{Zero-knowledge OTS proof system for $\QMA$} What is more, we show that the OTS proof system for $\QMA$ is also \emph{statistical zero-knowledge}, meaning that we can simulate in classical polynomial time the verifier's transcript when interacting with the provers on a yes-instance.

\begin{thm}(Restated in  \Cref{thm:QMA in Skep}) \label{thm:InformalZK}
For every language $L$ in $\QMA$, there exists a statistical zero-knowledge OTS proof system for $L$.
\end{thm}

\subsubsection{Delegated quantum computation in the OTS model}
As our final conceptual contribution, we show how our OTS proof system for $\QMA$ (\Cref{thm:InformalQMA}) can be adapted to the setting of delegated quantum computation (see \Cref{sec:intro-delegated}); note that the ZK property as described above also extends to the delegated quantum computation paradigm.

\begin{thm}(Restated as \Cref{thm:BQP Delegation}) \label{thm:InformalDelegated}
$\BQP$ has a relativistic delegated quantum computation protocol in the OTS model with the statistical zero-knowledge property. 
\end{thm}

 We believe that this  result is of particular impact since it addresses a new model for delegated quantum computation that has distinct conceptual benefits over existing delegation protocols:

\begin{enumerate}
\item Comparing our delegated quantum computation protocol to single-server protocols, we note that we make an extra assumption of an off-the-shelf, isolated device. However, the benefits are:
    \begin{enumerate}
\item     We achieve soundness against an unbounded server; existing single-server, classical-client delegation protocols require computational assumptions~\cite{Mah18}.
    \item The client does not trust \emph{any} quantum device at all; existing single-server, statistically secure protocols for delegated quantum computation require trust in a small quantum preparation device~\cite{Bro18,FK17}.
 \end{enumerate}
\item Comparing our delegated quantum computation protocol to existing multiple-server ($\MIP^*$) protocols, we note that:
 \begin{enumerate}
 \item Our approach only requires a single high-performance quantum server that handles the bulk of the computations; with a secondary efficient and generic device which need not even be given a description of the problem instance. This has practical advantages, especially when we consider that the off-the-shelf device can be acquired ahead of the verification stage (\Cref{fig:Delegated quantum computation}).
     \item Our approach is a single round, which means that relativistic means to enforce isolation are possible. The only other known relativistic protocol requires full quantum computational power for both servers and is not ZK~\cite{Gri19}.
 \end{enumerate}

\end{enumerate}

\subsection{Proof approach and technical contributions}
\label{sec:technical-contrib}

We now give an overview of the approach toward proving our main results, including an informal statement of our main technical contributions.

\subsubsection{Obstructions to the straightforward approach}

In delegating quantum computations in two- or multi-server models, the classical verifier is able to command quantum provers~\cite{RUV13} using two intertwined tests: (1) a computational test, with acceptance probability based on the required quantum computation (\emph{e.g.}, computation-by-teleportation~\cite{RUV13} or energy checking of a local Hamiltonian \cite{Ji17,Gri19}); (2) a rigidity test, ensuring provers' actions stay within a known range (\emph{e.g.}, self-test via CHSH game or Pauli braiding test). In order to establish the  ZK property, we must show that responses from the provers can be simulated using a classical probabilistic polynomial-time (PPT) device. Generally, approaches used for the rigidity test can be simulated in a straightforward way, hence the difficulty in obtaining ZK in this setting is in simulating the energy test. Furthermore, even if both tests are simulatable in isolation, this does not guarantee the ZK property since a malicious verifier may form question pairs  emanating from different tests, during  a single round.

Grilo~\cite{Gri19} presents a game $\mcG(H)$ determined by an ``XZ-type"\footnote{These are Hamiltonians where each local term $H_i$ is a real linear combination of tensor products of the Pauli-$X$ and Pauli-$Z$ operators.} Hamiltonian $H$. Honest provers for this game share suitably many EPR pairs, and one prover privately holds a ground state for $H$. The game $\mcG(H)$ combines an energy test with the Pauli braiding test~\cite{NV17,Vid18}. During the energy test, one prover reports measurement results of a randomly chosen term $H_i$ on their side of EPR pairs, and the other provides teleportation keys from a Bell basis measurement on the other EPR pairs and the ground state. Combining the energy test with the Pauli braiding test allows the verifier to ensure that provers share $n$ EPR pairs and that the required Pauli-X/Pauli-Z measurements are performed when measuring the local term~$H_i$.

The straightforward approach to obtaining a two-prover ZK proof system would be to combine recent results on simulatable codes in order to make the measurement results in Grilo's energy test simulatable. More specifically, one could apply the well-known circuit-to-Hamiltonian construction using the family of simulatable verification circuits given in~\cite{BG22}. Given such a circuit $V$, it is shown that local measurements on the ground state of the corresponding Hamiltonian $H_V$ are simulatable and thus this approach  would make the results of the energy test simulatable. Unfortunately, this approach fails for two technical reasons.

\paragraph{The choice of encoding.} Firstly, one cannot employ previously-known self-testing techniques to show the players perform the required measurements on the simulatable ground states given in~\cite{BG22}. On the one hand, previously-studied single-round self-testing techniques can only be used to show the players perform Pauli-$X$, and Pauli-$Z$ measurements. On the other hand, the choice of physical gates used by Broadbent and Grilo during the encoding of logic gates may result in a local Hamiltonian that is not of $XZ$-type and thus local terms $H_i$ may require measurements that have no known self-test.

\paragraph{The size of the measurement.} The second obstruction arises from the fact that existing rigidity tests in this setting require both players to make large-sized measurements on their shared state. 
 These large measurements can provide an avenue for attack by a malicious verifier which compromises the zero-knowledge property. In particular, since the Pauli braiding test allows for requests for measurements on all qubits, a malicious verifier may indicate to one player that an energy test is being played and simultaneously request Pauli-$X$ and \mbox{Pauli-$Z$} measurements on a large number of qubits. Such a measurement result cannot be simulated using simulatable codes, which only protect against constant-sized measurements, and thus this compromises zero-knowledge.

\subsubsection{Overview of proof and technical results}\label{sec:intro-overview-technical}

In order to correct for an appropriate choice of encoding, we prove that one can re-instantiate the verification circuit given by Broadbent and Grilo using an approach to simulatable codes given in \cite{GSY19}. This change allows us to encode logical gates of the verification circuit given by Broadbent and Grilo using a different set of physical gates and consequently, we show that the local Hamiltonian corresponding to the circuit is of $XZ$-type, while preserving simulatability.

\begin{thm}[Informal version of \Cref{thm:simulation of history states}]
    For any language $L=(L_{yes},L_{no})$ in $\QMA$, there is a family of verification circuits $V_{x}$ satisfying (1) the circuit-to-Hamiltonian construction applied to $V_x$  produces a Hamiltonian $H_x$ which is of $XZ$-type, and (2) if $x \in L_{yes}$ there exists a polynomial-time algorithm that can approximate the reduced density matrix obtained by tracing out all but 6 qubits of the ground state of $H_x$.
\end{thm}

In order to overcome the large measurement problem\footnote{In \Cref{subsec:OtherTests} we discuss other potential approaches to tackling the large measurement problem.}, we introduce a new self-test called the \emph{low-weight Pauli braiding test} (LWPBT) which can self-test the low-weight tensor products of Pauli measurements and $n$~EPR pairs but only requires the players to make measurements on a constant number of qubits.

\begin{thm}[Informal version of \Cref{rigidity}]\label{thminformal:rigidity}
The low-weight Pauli braiding test can self-test for $n$~EPR pairs and $6$-qubit Pauli measurements. This self-test is robust in the sense that any $\epsilon$-perfect strategy must be $poly(n)\sqrt{\epsilon}$ close to the canonical strategy.
\end{thm}

We use a group-theoretical approach to prove the rigidity of the LWPBT. Here, we briefly describe the challenges associated with achieving this result. Our proof involves approximate representations and approximate homomorphisms for groups. Informally, given a group presentation $\ang{S:R}$ for a group $G$, a map $f:S \rightarrow \msU$ from the generating set $S$ to the set of unitary operators $\msU$ is called an \emph{$\epsilon$-representation} if it respects the relations in $R$, up to error $\epsilon$. The more general notion of an \emph{$\epsilon$-homomorphism} refers to a map $f: G \rightarrow \msU$ that respects the multiplication of elements in the group, up to error $\epsilon$. When $\epsilon=0$, these two notions coincide. The well-known Gowers-Hatami theorem~\cite{GH17}  states that every approximate homomorphism of a finite group must be close to an exact homomorphism.

Implicitly in \cite{Vid18} for the rigidity of the $n$-qubit Pauli braiding test, a presentation is taken for the $n$-qubit Weyl-Heisenberg group which has $exp(n)$ generators subject to $exp(n)$ relations. It is then shown that any $\epsilon$-perfect strategy forms an $O(\sqrt{\epsilon})$-approximate representation, and given the large presentation one can straightforwardly show that this approximate representation is indeed an $O(\sqrt{\epsilon})$-approximate homomorphism.

The analysis of the LWPBT requires a presentation for the $n$-qubit Weyl-Heisenberg group with only $poly(n)$ generators and $poly(n)$ relations. Having fewer generators and relations makes the rigidity analysis more technical. In particular, to show that the approximate representation given by a near-perfect strategy is indeed an approximate homomorphism, we need to specify a normal form of the Weyl-Heisenberg group with respect to this ``small" presentation and find a rewriting procedure that takes arbitrary words of the group into its normal form. This allows us to track the precise error bounds on the approximate homomorphism we obtain and conclude that any $\epsilon$-perfect strategy forms a $poly(n)\sqrt{\epsilon}$-approximate homomorphism.

In order to finally round an approximate homomorphism to an exact homomorphism, we make further improvements to the state-of-the-art understanding of the stability of finite groups. In particular, in \Cref{thm:GH} we state and prove an enhanced version of the Gowers-Hatami theorem that can be used for the stability analysis of the Weyl-Heisenberg group.
Aside from our use case, this new version can simplify previous approaches to self-testing. In brief, when constructing an exact homomorphism from an approximate homomorphism, our enhanced version allows one to disregard irrelevant  sub-representations without truncating the isometry given by the Gower-Hatami theorem. This improvement can help clarify some subtle issues since, in general, truncation of an isometry may fail to be an isometry.

\begin{thm}[Informal version of \Cref{thm:GH}]
    If $f$ is an approximate homomorphism of a finite group $G$ on some Hilbert space $\mcH$, then there is a Hilbert space $\mcK$, an isometry $V:\mcH\arr\mcK$, and an exact homomorphism $\phi$ of $G$ on $\mcK$ such that $V^*\phi V$ is close to $f$. If in addition, $f$ restricts to a representation on a given subgroup $S$ of $G$, and an irreducible representation $\xi$ has zero Fourier coefficient in $f|_S$, then $\xi$ has zero support in $\phi$.
\end{thm}

We use the above technical results to derive a modified version of~\cite{Gri19} by interleaving the following tests: (1) a computational test consisting of an energy test in which a simulatable witness uses low-weight Pauli-$X$ and Pauli-$Z$ measurements and, (2) a rigidity test consisting of the LWPBT. The result of this modified Grilo protocol gives a ZK OTS protocol with an inverse polynomial completeness-soundness gap. Finally, we apply a threshold parallel repetition theorem to the above protocol to amplify the completeness-soundness gap to be constant, thus demonstrating both~\Cref{thm:InformalQMA} and \Cref{thm:InformalZK}. We then show that the proof system is of a form that can be scaled down to yield a delegation protocol, yielding~\Cref{thm:InformalDelegated}.

\subsection{Open problems}\label{subsec:OpenProblems}

In this work, we show that the OTS model displays many interesting properties that are related to both the single- and multi-prover models for interactive proofs.  We believe that there is tremendous scope for further investigation of the model, and also that our novel techniques will find applications elsewhere.   We collect some open problems  below.

\textbf{Do both provers need to be all-powerful to recover $\MIP^*$?} Our results show that $\QMA \subset \OTS$.  We leave it as an open question to determine if the same can be said about larger complexity classes such as $\QIP$. Going further, perhaps one can characterize how powerful the untrusted set-up device needs to be in order to recover all of $\MIP^*$. More specifically, if we allow the untrusted device to make arbitrary measurements but still require independence from the input, would we be able to recover $\MIP^*$? Somewhat dual to this question would be to ask if one can prove a natural upper bound for our new model. One plausible candidate would be the class $\QIP$ but even this is not immediately straightforward if only for the reason that the off-the-shelf device may share arbitrary amounts of entanglement with the prover.

\textbf{Applications to proofs of quantum knowledge.} As an extension to the concept of ZK, in a \emph{proof of knowledge (PoK)}, the verifier further becomes convinced that the prover has ``knowledge'' of an accepting witness. This is formally treated by showing the existence of an efficient \emph{knowledge extractor} which outputs a witness given oracle access to the prover; this concept has recently been extended to the quantum setting \cite{BG22}, \cite{CVZ19arxiv}. We conjecture that the assurances given by self-testing provide a different approach to convincing the verifier that the prover(s) are in possession of an actual quantum witness; it may be possible to refine our approach to self-testing by showing the result for a more complete game. Namely, we leave it as an open question to find a 2-prover, 1-round protocol that self-tests for ground states of a local Hamiltonian, and to determine how this yields a new approach to a quantum PoK.

\textbf{OTS devices for delegated quantum computation.} We have shown that the OTS model is powerful enough for the delegated quantum computations of~$\BQP$;
while it seems unlikely that we can do away with the OTS device completely~\cite{ACGK19}, an open question is to determine the \emph{minimal} power required by this device (assuming $\BQP \neq \BPP$).

\textbf{Further uses and refinements of the LWPBT.} We believe the new LWPBT may be more widely applicable since it may be further generalized to provide a straightforward tradeoff of question/answer size for robustness in a self-test for EPR pairs and Pauli measurements. For example, one may analyze the same game as in our main result, but allow for measurements of size $\log(n)$ instead of constant. We leave for future work the study of this model and its consequences.

\textbf{Further Applications.} There are a variety of interesting applications and considerations made of interactive proof systems in the quantum setting aside from the ones that we studied.
Recent works have introduced the notion of state/unitary complexity classes \cite{RY22}. Informally, the approach is to study the ability of an efficient quantum verifier to output a target quantum state, given access to an untrusted quantum prover. Similar considerations can be made in our OTS model, in which a classical verifier may want to certify the existence of a target state on the register of the untrusted set-up device.

\subsection{Acknowledgements}
We thank Seyed Sajjad Nezhadi, Gregory Rosenthal, William Slofstra, Jalex Stark, and Henry Yuen for discussions on the model. We thank Alex Grilo for discussions on some of our proof techniques. We also thank Thomas Vidick for detailed discussions on the formulation and proof of \Cref{thm:GH}. We thank the anonymous reviewers for suggesting the use case of an OTS for a classical proof system.

This work was supported by the Mitacs Accelerate program IT24833 in collaboration with industry partner Agnostiq, online at \url{agnostiq.ai}. A.B.~is supported by the Air Force Office of Scientific Research under award number FA9550-20-1-0375,  NSERC, and the University of Ottawa’s Research Chairs program.

\subsection{Outline}

The remainder of this paper is organized as follows. We give needed preliminaries in \Cref{sec:prelims}. In \Cref{Sec:LWPBT}, we introduce our low-weight Pauli braiding test and we prove the rigidity properties of this test in \Cref{sec:rigidity}. \Cref{sec:rigidity} also contains the statement and proof of our enhanced version of the Gowers and Hatami stability result. In \Cref{sec:Game}, we provide a modified version of the Hamiltonian game first introduced by Grilo \cite{Gri19} and prove the completeness and soundness properties of this game. In \Cref{Sec:ZK}, we provide some background on simulatable codes and apply them to the game outlined in \Cref{sec:Game} in order to determine the OTS proof system for~$\QMA$. In \Cref{sec:SkepticalVerifierModel}, we introduce and formally define our OTS  model, and show how the technical contributions of the prior sections come together to show all three conceptual results.

%% file: Preliminaries.tex
We introduce the notation and provide a background discussion on topics including representations theory of groups, non-local games, and zero-knowledge ${\MIP}^*$ protocols. We also define several relevant complexity classes.

\subsection{Notation}

We take $[n]$ to denote the set $\{1, \dots n \}$. Given two real valued functions $f,g:\R\rightarrow\R$, we write $f=O(g)$ (resp. $f=\Omega(g)$) if there exists a positive real number $M$ and an $x_0\in \R$ such that $\abs{f(x)}\leq Mg(x)$ (resp. $\abs{f(x)}\geq Mg(x)$) for all $x\geq x_0$. We call a function $f$ negligible, and write $f=\negl(n)$, if for all constants $c>0$ we have $f = O(n^{-c}).$ For two distributions $P$ and $Q$ on a finite set $\mathcal{X}$ the statistical differences of $P$ and $Q$ is given by $ \sum_{x\in\mathcal{X}} \abs{P(x)- Q(x)}$.

In this paper, all Hilbert spaces are finite-dimensional. Given a Hilbert space $\mcH$, we use $\msB(\mcH)$ to denote the set of bounded linear operators acting on $\mcH$, use $\msU(\mcH)$ to denote the group of unitary operators on $\mcH$, and use $\Id_{\mcH}$ to denote the identity operator on $\mcH$. Given an operator $A \in \msB(\mcH)$ we take $A^*$ to denote the adjoint operator (equivalently the conjugate transpose) and define the trace norm $\|A \|_{tr} := \mathrm{Tr}\sqrt{A^*A}$.

\subsection{Quantum information}
A quantum state $\rho$ on $\mcH$ is a positive operator in $\msB(\mcH)$ with $\Tr(\rho)=1$. It induces a semi-norm $\norm{A}_\rho:=\sqrt{\norm{A^*A\rho}}$ on $\msB(\mcH)$ which we call the $\rho$-norm. This norm is left unitarily invariant, meaning that $\norm{UA}_\rho=\norm{A}_\rho$ for all $U\in \msU(\mcH)$ and $A\in\msB(\mcH)$. Given two quantum states $\rho$ and $\sigma$ we define their trace distance $D(\rho, \sigma) = \frac{1}{2}\|\rho - \sigma \|_{tr} = \max_{P} \mathrm{Tr}(P(\rho - \sigma))$ where the max is taken over all projections $P \in \msB(\mcH)$.

We use $\ket{\epr}$ to denote the EPR pair in $\C^2\otimes\C^2$ and use $\ket{\eprn}$ to denote the $n$-qubit EPR pair. We also take $\sigma_I $, $\sigma_X$, and $\sigma_Z$ to denote the following Pauli operators:
\begin{equation}
\sigma_I = \begin{bmatrix}
        1 & 0\\ 0 & 1
    \end{bmatrix}, \,
    \sigma_X = \begin{bmatrix}
        0 & 1\\ 1 & 0
    \end{bmatrix}, \text{ and }
    \sigma_Z = \begin{bmatrix}
        1 & 0 \\ 0 & -1
    \end{bmatrix}.
\end{equation}
 For every $a \in\{0,1\}^n$ and $W \in \{ I, X, Z \}^n$, we use $\sigma_W(a)$ to denote the operator $\otimes_{i \in [n]} \sigma_{W_{i}}^{a_{i}}$ on $(\C^2)^{\otimes n}$ where $\sigma_I^0=\sigma_X^0=\sigma_Z^0=\sigma_I$. Definitions of these gates and other fundamental concepts from quantum computing can be found in \cite{NC00}.

\paragraph{Families of Quantum Circuits}

A \emph{unitary quantum circuit} is simply a unitary which can be written as a product of gates from some universal gate set $\mathcal{U}$. Unless otherwise specified we will assume the universal gate set is the following universal gate set $\lbrace H, \Lambda(X), \Lambda^{2}(X) \rbrace$, where $H$ is the Hadamard gate, $\Lambda(X)$ is the controlled $\sigma_X$ gate, and $\Lambda^{2}(X)$ is the Toffoli gate. A \emph{general quantum circuit} or simply a \emph{quantum circuit} is a unitary quantum circuit that can additionally apply non-unitary gates which, introduce qubits initialized in the $0$ state, trace out qubits, or measure qubits in the standard basis.

\begin{defn}[Polynomial-time uniform circuit family]
We say a family of quantum circuits $\{ Q_n \}_{n \in \mathbb{N}} $ is a polynomial-size family of quantum circuits if there exists polynomial $r$ such that $Q_n$ has size at most $r(n)$. A family of  quantum circuits $\{ Q_n \} $ is called \emph{polynomial-time uniform} family if there exists a polynomial time Turing machine that on input $1^n$ outputs a description of $Q_n$. In this case, the family will also be a polynomial-size family of quantum circuits.
\end{defn}

Given a quantum circuit $Q$, we denote its size (number of gates and number of wires) by~$|Q|$.  The task of delegating the computation of $Q$ is captured by the following promise problem:
\begin{defn}[Q-CIRCUIT]\label{defn:Q-circuitProblem}
The input is a quantum circuit $Q$ on $n$ qubits. The problem is to decide between the following two cases:
\begin{itemize}
\item \textbf{Yes.} $\lVert((|1\rangle\langle1|\otimes I_{n-1})Q|0^n\rangle\rVert^2 \geq 1-\gamma$
\item \textbf{No.}  $\lVert((|1\rangle\langle1|\otimes I_{n-1})Q|0^n\rangle\rVert^2 \leq \gamma$
\end{itemize}
when we are promised that one of the two cases holds.
\end{defn}

Problem in \Cref{defn:Q-circuitProblem} is known to be $\BQP$-complete for $1-2\gamma > \frac{1}{\mathrm{poly}(n)}$.

\subsection{Groups and representations}\label{subsec:approxrep}

In this paper, we work with groups and their (approximate) representations.
We use $1$ for the identity in groups, and use $[g,h]:=ghg^{-1}h^{-1}$ to denote the group commutator. A group is said to be \emph{finite} if it contains finitely many elements. A \emph{(unitary) representation} $\phi:G\rightarrow \msU(\mcH)$ of a group $G$ on a Hilbert space $H$ is a group homomorphism from $G$ to $\msU(\mcH)$. A subspace $\mcK\subset\mcH$ is said to be an \emph{invariant subspace of $\phi$} if $\phi(G)\mcK:=\{\phi(g)\ket{k}:g\in G,\ket{k}\in\mcK\}=\mcK$. A representation is \emph{irreducible} if it has no proper non-zero invariant subspace. Let $\phi_1:G\arr\msU(\mcH_1)$ and $\phi_2:G\arr\msU(\mcH_2)$ be two representations of a group $G$. We say $\phi_1$ and $\phi_2$ are \emph{unitarily equivalent} if there is a unitary $U:\mcH_1\arr \mcH_2$ such that $U\phi_1(g)U^*=\phi_2(g)$ for all $g\in G$. The \emph{direct sum} of $\phi_1$ and $\phi_2$, denoted $\phi_1\oplus\phi_2$, is a representation of $G$ on $H_1\oplus H_2$ sending $g\mapsto \phi_1(g)\oplus\phi_2(g)$. Maschke's theorem states that every representation $\phi$ of a finite group $G$ is unitarily equivalent to a direct sum of irreducible representations $\{\phi_i:1\leq i\leq k\}$ of $G$. In this case, we say $\bigoplus_{i=1}^k\phi_i$ is the \emph{irreducible decomposition} of $\phi$, and every $\phi_i$ is an \emph{irreducible component} of $\phi$.

For a finite group G, we use $\Irr(G)$ to denote the unique (up to unitary equivalence of elements) complete set of inequivalent irreducible representations. We assume without loss of generality that every $\varphi\in \Irr(G)$ is a representation on $\C^{d_{\varphi}}$ with the standard basis $\{\ket{i}:i\in [d_\varphi]\}$. It is well-known that
\begin{equation*}
   \sum_{\varphi\in \Irr(G)}d_{\varphi}\Tr\big(\varphi(g)\big)=\begin{cases}  \abs{G} & \text{if } g=1, \\
0 & \text{otherwise.}
\end{cases}
\end{equation*}

Given a finite group $G$, a function $f:G\arr\msU(\mcH)$, and an irreducible representation $\phi:G\arr\msU(\C^{d})$, the \emph{Fourier transform of $f$ at $\phi$} is an operator
\begin{align}
    \hat{f}(\phi):=\frac{1}{\abs{G}}\sum_{g\in G}f(g)\otimes \overline{\phi(g)}
\end{align}
acting on $\mcH\otimes\C^d$, where $\overline{\phi(g)}$ is the conjugate of the matrix $\phi(g)\in M_{d}(\C)$ in the standard basis.

Let $f:G\arr \msU(\mcH)$ be a function of a finite group $G$. Given a quantum state $\rho$ on $\mcH$ and a positive real number $\epsilon$, we say $f$ is an \emph{$(\epsilon,\rho)$-homomorphism} provided that $f(g^{-1})=f(g)^*$ and $\frac{1}{\abs{G}}\sum_{h\in G}\norm{f(g)f(h)-f(gh)}_{\rho}^2\leq \epsilon$ for all $g\in G$. In this case, by the well-known Gowers-Hatami theorem \cite{GH17,Vid18}, there is a Hilbert space $\mcK$, an isometry $V:\mcH\arr\mcK$, and a representation $\phi:G\arr\msU(\mcK)$ such that $\norm{f(g)-V^*\phi(g)V}_\rho\leq \epsilon$ for all $g\in G$. In \Cref{GH} we introduce an enhanced version of this theorem which allows us to disregard all one-dimensional irreducible representations of the Weyl-Heisenberg group. Earlier works dealt with these one-dimensional representations by invoking a truncation of the isometry given by the Gowers-Hatami theorem. Unfortunately, in general, truncation of an isometry can fail to be an isometry.

We also work with group presentations. Given a set $S$, we use $\mathcal{F}(S)$ to denote the free group generated by $S$, and use $\ang{S:R}$ to denote the quotient of $\mathcal{F}(S)$ by the normal subgroup generated by $R$. We say $\ang{S:R}$ is generated by $S$ subject to the relations $R$. When $S$ and $R$ are both finite sets, we say the group $\ang{S:R}$ is \emph{finitely-presented}. Given a group $G=\ang{S:R}$, a \emph{normal form} for $G$ with respect respect to $S$ is an injective function $\mcN$ from $G$ to words over $S\cup S^{-1}$ such that $\mcN(g)=g$ in $G$ for all $g\in G$.

\subsection{Non-local games and rigidity}\label{subsec:rigidity}
A two-player\footnote{These two players are commonly called Alice and Bob.} one-round nonlocal game $\mcG$ is a tuple $\big(\lambda,\mu,\mcI_A,\mcI_B,\mcO_A,\mcO_B\big)$ , where $\mcI_A,\mcI_B$ are finite input sets, and $\mcO_A,\mcO_B$ are finite output sets, $\mu$ is a probability distribution on $\mcI_A\times\mcI_B$, and $\lambda:\mcO_A\times\mcO_B\times\mcI_A\times\mcI_B\rightarrow\{0,1\}$ determines the win/lose conditions. A quantum strategy $\mcS$ for~$\mcG$
is given by finite-dimensional Hilbert spaces $\H_A$ and $\H_B$, a unit vector $\ket{\psi}\in\H_A\otimes \H_B$, Alice's POVMs $\{E_a^x:a\in\mcO_A \},x\in\mcI_A$ on $\H_A$, and Bob's POVMs $\{F_b^y:b\in\mcO_B\},y\in\mcI_B$ on $\H_B$. The winning probability of $\mcS$ for game $\mathcal{G}$ is given by
\begin{align*}
\omega(\mcG, \mcS):=\sum\limits_{a,b,x,y}\mu(x,y)\lambda(a,b\vert x,y)\bra{\psi}E_a^x\otimes F_b^y\ket{\psi}.
\end{align*}
A quantum strategy $\mcS$ for a non-local game $\mcG$ is said to be \emph{perfect} if $\omega(\mcG, \mcS)=1$. When the game is clear from the context we simply write $\omega(\mcS)$ to refer to the winning probability. The \emph{quantum value} of a non-local game $\mcG$ is defined as
 \begin{align*}
     \omega^*(\mcG):=\sup\{\omega(\mcS):\mcS \mbox{  a quantum strategy for } G\}.
 \end{align*}

In this paper, we assume all measurements employed in a quantum strategy are PVMs. An $m$-outcome PVM $\{P_1,\cdots,P_m\}$ corresponds to an observable $\sum_{j\in[m]}\exp(\frac{2\pi i}{m}j)P_j$, so a quantum strategy for a game $\mcG=\big(\lambda,\mu,\mcI_A,\mcI_B,\mcO_A,\mcO_B\big)$ can also be specified by a triple
\begin{align*}
    \mcS=(\tau^A,\tau^B,\ket{\psi}\in\mcH_A\otimes\mcH_B)
\end{align*}
where $\tau^A(x)$, $x\in \mcI_A$ are $\mcO_A$-outcome observables on $\mcH_A$, and $\tau^B(y)$, $y\in \mcI_B$ are $\mcO_B$-outcome observables on $\mcH_B$.

Here we introduce the well-known Mermin-Peres Magic Square game, in which Alice and Bob are trying to convince the verifier that they have a solution to a system of equations over $\mathbb{Z}_2$. There are 9 variables $v_1, \dots, v_9$ in a $3\times 3$-array whose rows are labeled $r_1,r_2,r_3$ and columns are labeled $c_1, c_2, c_3$.
\begin{table}[ht]
    \centering
   \begin{tabular}{c|c|c|c|}
\nolines{} & \nolines{$c_1$} & \nolines{$c_2$} & \nolines{$c_3$}  \\ \cline{2-4}
$r_1$ & $v_1$ & $v_2$ & $v_3$  \\ \cline{2-4}
$r_2$ & $v_4$ & $v_5$ & $v_6$ \\ \cline{2-4}
$r_3$ & $v_7$ & $v_8$ & $v_9$ \\ \cline{2-4}
\end{tabular}
    \caption{Magic square game}
\end{table}

Each row or column corresponds to an equation: variables along the rows or columns in $\{r_1,r_2,r_3,c_1,c_2\}$ sum to 0; variables along the column $c_3$ sum to 1. In each round, Bob receives one of the 6 possible equations and he must respond with a satisfying assignment to the given equation. Alice is then asked to provide a consistent assignment to one of the variables contained in the equation Bob received. The following table describes an operator solution for this system of equations:
\begin{table}[h]
    \centering
    \begin{tabular}{lll}
    $A_1= \sigma_I \otimes \sigma_{Z}$ &  $A_2 = \sigma_{Z} \otimes \sigma_I$ & $A_3 = \sigma_{Z} \otimes \sigma_{Z}$ \\
     $A_4 = \sigma_{X} \otimes \sigma_I$ & $A_5 = \sigma_I \otimes \sigma_{X}$ & $A_6 = \sigma_{X} \otimes \sigma_{X}$ \\
     $A_7 = \sigma_{X} \otimes \sigma_{Z}$ & $A_8 = \sigma_{Z} \otimes \sigma_{X}$ & $A_9 = \sigma_X\sigma_Z \otimes \sigma_Z\sigma_X$
\end{tabular}
\label{Alice}
\caption{Operator solution for Magic Square game}
\end{table}

The canonical perfect quantum strategy for this game is one in which
\begin{itemize}
    \item the players share two EPR pairs,
    \item given a variable $v_i$, Alice performs $A_i$ on her registers, and
    \item given a row or column consisting of three variables $v_j,v_l$, and $v_\ell$, Bob perform $A_jA_kA_\ell$ on his registers.
\end{itemize}

\begin{defn}
Let $\mcS=(\tau^A,\tau^B,\ket{\psi}\in\mcH_A\otimes\mcH_B)$ and $\widetilde{\mcS}=( \{\widetilde{\tau}^A \},\{\widetilde{\tau}^B \},\ket{\widetilde{\psi}}\in\widetilde{\mcH}_A\otimes\widetilde{\mcH}_B )$ be two quantum strategies for a game $\mcG=\big(\lambda,\mu,\mcI_A,\mcI_B,\mcO_A,\mcO_B\big)$. We say  $\mcS$ is $\delta$-close to  $\widetilde{\mcS}$, written $\mcS\succeq_{\delta}\widetilde{\mcS}$, if there are Hilbert spaces $\mcH_A^{aux}$ and $\mcH_B^{aux}$,  isometries $V_A:\mcH_A\rightarrow \widetilde{\mcH}_A\otimes\mcH_A^{aux}$ and $V_B:\mcH_B\rightarrow\widetilde{\mcH}_B\otimes\mcH_B^{aux}$, and a unit vector $\ket{aux}\in\mcH_A^{aux}\otimes\mcH_B^{aux}$ such that
\begin{align}
    \norm{ (V_A\otimes V_B)(\tau^A(x)\otimes \tau^B(y)\ket{\psi}) - (\widetilde{\tau}^A(x)\otimes\widetilde{\tau}^B(y)\ket{\widetilde{\psi}})\otimes\ket{aux} }^2\leq\delta
\end{align}
for all $(x,y)\in\mcI_A\times\mcI_B$.
\end{defn}

The rigidity of the Magic Square game has been well studied \cite{WBMS16}:
\begin{lemma}\label{lemma:MS}
If $\mcS$ is a strategy for the Magic Square game with winning probability $1-\epsilon$, then $\mcS$ is $O(\sqrt{\epsilon})$-close to the canonical perfect strategy.
\end{lemma}

\subsection{Complexity classes and zero knowledge}\label{Sec:PrelimsZk}

\begin{defn}[$\QMA$]\label{defn:QMA}
A promise problem $L=(L_{yes},L_{no})$ is in $\QMA$  if there exist
polynomials $p$ and $q$, and a polynomial-time uniform family of quantum circuits
$\{ Q_n \}$ where $Q_n$ takes as input a string $x\in\Sigma^*$ with
$|x|=n$, a $p(n)$-qubit quantum state $\ket{\psi}$, and $q(n)$ auxiliary qubits in state
$\ket{0}^{\otimes q(n)}$, such that:
  \begin{itemize}
    \item (Completeness)  if $x\in L_{yes}$, then there exists some $\ket{\psi}$
      such that $Q_n$ accepts $(x,\ket{\psi})$ with probability at least
      $1-\negl(n)$, and
   \item (Soundness) if $x\in L_{no}$, then for any state $\ket{\psi}$, $Q_n$
      accepts $(x,\ket{\psi})$ with probability at most $\negl(n)$.
  \end{itemize}
\end{defn}
We sometimes refer to the family of circuits $\{ Q_n \}$ in \Cref{defn:QMA} simply as a family of \emph{verification circuits}.

Informally, a language is in $\MIP^*$ if it can be decided by a polynomial-time classical verifier that interacts with \emph{multiple} isolated, all-powerful and entangled provers. In this work, we only deal with $\MIP^*$ proof systems involving two provers and one round; which are defined as the following.

\begin{defn}
  A promise language $L=(L_{yes}, L_{no})$ is in $\MIP^*[2,1]_{c,s}$ if there exists a polynomial-time computable function that takes an instance $x \in L$ to a description of a non-local game $\mcG_x$ satisfying the following conditions.
  \begin{itemize}
      \item (Completeness) For every $x \in L_{yes}$ we have $\omega^*(\mcG_x) \geq c$.
      \item (Soundness) For every $x \in L_{no}$ we have $\omega^*(\mcG_x) < s$.
  \end{itemize}
\end{defn}
We refer to the mapping, $x \mapsto \mcG_x$, as a $\MIP^*[2,1]_{c,s}$ proof system, or in some places a $\MIP^*[2,1]_{c,s}$ protocol. When the parameters are clear from the context we simply call it an $\MIP^*$ proof system.

Next, we discuss \emph{zero knowledge}, which is an additional property for proof systems. Intuitively, in a zero-knowledge proof system, no verifier (including a malicious one) can learn anything beyond the membership in the language of the instance under consideration.
 For our scenario of  $\MIP^*[2,1]$, it is sufficient to quantify over verifiers that send at most one valid question to each prover (otherwise, we can impose that the provers abort the protocol).   Hence,  a malicious verifier can only deviate from the honest one by sampling an initial question according to an alternate distribution, and by sampling the second question adaptively.

More formally, a malicious verifier $\widehat{V}$ is a probabilistic polynomial-time Turing machine which on input $x$ and randomness $\theta$ samples question $q_1$ for either Alice or Bob. Given reply $r_1$, the malicious verifier samples question $q_2$ in a way that may depend on $q_1$ and $r_1$.  For a given quantum strategy $\mathcal{S}$ and malicious verifier $\widehat{V}$, we take $View(\widehat{V}(x), \mathcal{S})$ to be the random variable corresponding to the transcript of questions and answers $(x,\theta, q_1,r_1,q_2,r_1)$. A protocol is zero-knowledge when for all ``yes" instances a simulator can sample from the distribution above.

\begin{defn}\label{defn:SZKMIP*}
    An $\mathrm{MIP}^*[2,1]_{c,s}$ proof system is \emph{statistical zero-knowledge} if for every $x \in L_{yes}$ there exists an honest prover strategy $\mathcal{S}$ satisfying the following:
    \begin{enumerate}
        \item $\omega^*(\mathcal{S}) \geq c$.
        \item For any PPT malicious verifier $\widehat{V}$ there exists a PPT simulator $Sim_{\widehat{V}}$ with output distribution that is $\epsilon$-close to $View(\widehat{V}(x),S)$ in statistical distance for some negligible function $\epsilon(|x|)$.
    \end{enumerate}
\end{defn}

\subsection{Simulatable codes and encodings of gates}

Recall that a quantum error-correcting code (QECC) $\mathcal{C}=[[n,k]]$ is a map $\mathrm{Enc}:(\C^2)^{\otimes k}\rightarrow (\C^2)^{\otimes n}$, which encodes a $k$-qubit state $\ket{\psi}$ into an $n$-qubit state $\mathrm{Enc}(\ket{\psi})$ where $n\geq k$. The code is said to have distance $d$ if the original state can be recovered from the encoded state that has transformed under any quantum operation which acts on at most $(d-1)/2$ qubits. Given an $[[m,1]]$ QECC with map $\mathrm{Enc}$, we abuse notation and also write $\mathrm{Enc}$ for the corresponding $[[mn,n]]$ encoding that is obtained by applying $\mathrm{Enc}$ to each of the qubits in an $n$-qubit system.

We use $\underline{A}^k_n$ to denote the set of $k$ distinct numbers between $1$ and $n$ through this section. Then $\underline{A}^k:=\bigcup_{n\geq k}\underline{A}^k_n$ is the set of $k$ distinct numbers. Given a $k$-qubit logical gate $U$ and an element $\underline{a} =(a_1, \dots, a_k) \in \underline{A}^k$, let $U(\underline{a})$ denote the gate $U$ applied to qubits $a_1,\cdots a_k$.

Below we recall the definition of simulatable codes, which was first introduced by Grilo, Slofstra, and Yuen \cite{GSY19}.

\begin{defn}
Given a $k$-qubit logical gate $U$ and a quantum error-correcting code $\mathcal{C}=[[m,1]]$, let $(\sigma_U, \sigma'_U)$ be a pair of states, and let $\ell$ be a positive integer. For each $1\leq i\leq \ell$, let $\mathcal{O}_i$ be a mapping from elements $\underline{a}=(a_1,\cdots,a_k)$ in $\underline{A}^k$ to unitaries $\mathcal{O}_i(\underline{a})$ acting only on
\begin{enumerate}[label=(\roman*)]
    \item the physical qubits of codewords in $\mathcal{C}$ that corresponds to logical qubits $a_1,\cdots,a_k$, and
    \item the register that holds $\sigma_U$.
\end{enumerate}
We say the tuple $(\sigma_U,\sigma'_U,\ell,\mathcal{O}_1,\cdots,\mathcal{O}_\ell)$ is an encoding of $U$ in code $\mathcal{C}$ if
\begin{align}
    (\mathcal{O}_{\ell}(\underline{a}) \dots  \mathcal{O}_1(\underline{a}))(Enc(\rho) \otimes \sigma_U)(\mathcal{O}_{\ell}(\underline{a}) \dots  \mathcal{O}_1(\underline{a}))^* =Enc\big(U(\underline{a})\rho U(\underline{a})^*\big) \otimes \sigma'_U
\end{align}
for all $n\geq k$, elements $\underline{a}\in \underline{A}^k_n$, and $n$-qubit states $\rho$.
If in addition, the unitaries $\mathcal{O}_1(\underline{a}),\cdots,\mathcal{O}_\ell(\underline{a})$ are gates in some set $\mathcal{U}$ for all $\underline{A}\in \underline{A}^k$, then we say the encoding $(\sigma_U,\sigma'_U,l,\mathcal{O}_1,\cdots,\mathcal{O}_\ell)$ uses physical gates in $\mathcal{U}$.
\end{defn}

Given a circuit of logical gates $V=U_1 \dots U_k$ we refer to an encoding of $V$ as the corresponding circuit of physical gates obtained by applying an encoding of each gate $U_i$.

\begin{defn}
An encoding $(\sigma_U, \sigma'_U, \mathcal{O}_1, \dots, \mathcal{O}_\ell)$ of a $k$-qubit logical gate $U$ in a QECC $\mathcal{C}$ is called $s$-simulatable if for all $0 \leq t \leq \ell$, $n$-qubit states $\rho$, and subsets $S$ of the physical qubits of $Enc(\rho) \otimes \sigma_U$ with $|S| \leq s$, the partial trace
\begin{align*}
   \mathrm{Tr}_{\overline{S}}\Big(\mathcal{O}_{t}(\underline{a}) \dots  \mathcal{O}_1(\underline{a}))(Enc(\rho) \otimes \sigma_U)(\mathcal{O}_{t}(\underline{a}) \dots  \mathcal{O}_1(\underline{a}))^*\Big)
\end{align*}
is a $2^{\abs{S}}\times 2^{\abs{S}}$ matrix whose entries are rational and can be computed in polynomial time from $t$, $\underline{a}$ and $S$. In particular, this matrix is independent of $\rho$ if $\mathcal{C}$ can correct arbitrary errors on s qubits.
\end{defn}

\begin{thm}[Theorem 6 in \cite{GSY19}]\label{thm:sim-codes}
Let $\mathcal{U}= \lbrace H, \Lambda(X), \Lambda^2(X) \rbrace $. For every $s\in \N$, there exists a constant $n\in \N$ and a $[[n,1]]$ QECC $\mathcal{C}$ such that any logical gate in $\mathcal{U}$ has an $s$-simulatable encoding in $\mathcal{C}$ using physical gates in $\mathcal{U}$.
\end{thm}

\subsection{Local Hamiltonains}
We define the \emph{local Hamiltonian problem}.

\begin{defn}
  \label{def:local Hamiltonian}
  Let
  $k \in \mathbb{N}$, $\alpha, \beta \in \mathbb{R}$ with $\alpha < \beta$, the $k$-\emph{Local Hamiltonian} problem with parameters $\alpha$ and $\beta$
  is the following promise problem.
  Let $n$ be the number of qubits of a quantum system.
  The input is a set of $m(n)$ Hamiltonians $H_1, \ldots, H_{m(n)}$
  where $m$ is a polynomial in $n$ and each $H_i$ acts on $k$ qubits out of the $n$ qubit system with $\norm{H_i} \leq 1 $.
  For $H = \sum_{j = 1}^{m(n)} H_j$ the promise problem is to decide between the following.
    \begin{itemize}
\item[\textbf{Yes.}] There exists an  $n$-qubit state $\ket{\varphi}$ such that
      $
        \bra{\varphi} H \ket{\varphi}
        \leq a \cdot m(n) .
      $
\item[\textbf{No.}] For every $n$-qubit state $\ket{\varphi}$
      it holds that
      $
        \bra{\varphi} H \ket{\varphi}
        \geq b \cdot m(n) .
      $
    \end{itemize}
\end{defn}

The above problem plays a central role in quantum proofs since it was shown to be $\QMA$ complete for parameters  $k=5$ and $\beta - \alpha = \frac{1}{\mathrm{poly}(n)}$ \cite{KSV02}; the proof of this fact uses the famous circuit-to-Hamiltonian construction.

%% file: LWPBT.tex
In this section, we introduce a non-local game that can self-test for Alice making measurements with $n$-qubit ``braids'' of Pauli $\sigma_X$ and $\sigma_Z$ of weight at most $6$, against $n$ EPR pairs. That is, Alice performs a measurement of the form $\otimes_{i} \sigma_{W_i}$, where each $W_i \in \{X,Z, I \}$ and for all but $6$ indices we have $W_i=I$. As we discuss in \Cref{sec:technical-contrib} we also want this game to be able to be won perfectly when Alice is restricted to only perform measurements on at most $6$ qubits. Consequently, we cannot use the typical Pauli Braiding Test. To this end, we introduce a low-weight version of the Pauli Braiding test. We introduce this game for weight $6$ but the rigidity arguments in the following section equally hold for the analogous game with any constant weight on the strings. This game is constructed from two sub-games which we call the low-weight linearity test and the low-weight anti-commutation test.

\subsection{Low-weight linearity test}
\label{sec:LW-commutation}

For any $a\in\{0,1\}^n$ and $W\in\{X,Z\}^n$, we use $W(a)$ to denote the sequence $W_1^{a_1}W_2^{a_2}\cdots W_n^{a_n}$ where $X^0=Z^0=I$. Let $\mcI_A:=\{W(a):W\in\{X,Z\}^n,a\in\{0,1\}^n \text{ such that }\abs{a}\leq 6\}$ and let $\mcI_B:=\{(W(a),W(a')):W\in\{X,Z\}^n,a,a'\in\{0,1\}^n \text{ such that }\abs{a},\abs{a'}\leq 6\}$ be the question sets for Alice and Bob respectively. We first describe the low-weight linearity test in \Cref{fig:(LW)LinearityTest}.

\begin{figure}[!htbp]
    \centering
\begin{mdframed}
   \begin{enumerate}
    \item  The verifier selects uniformly at random $W \in \lbrace X, Z \rbrace^{n}$  and strings $a, a' \in \lbrace 0,1 \rbrace^{n}$ satisfying $|a|, |a'| \leq 6$ (\emph{i.e.} $a,a'$ both have at most $6$ non-zero entries).
    \item The verifier sends $(W(a), W(a'))$ to Bob. If $a+a'$ has weight at most $6$ then the verifier selects $W' \in \lbrace W(a), W(a'), W(a+a') \rbrace$ uniformly at random to send to Alice. Otherwise, the verifier uniformly at random sends $W' \in \lbrace W(a), W(a') \rbrace$ to Alice.
    \item The verifier receives two bits $(b_1,b_2)$ from Bob and one bit $c$ from Alice.
    \item If Alice receives $W(a) $ then the verifier requires $b_1=c$. If Alice receives $W(a')$ then the verifier requires $b_2=c$. If Alice receives $W(a+a')$ then the verifier requires $b_1 +b_2 =c$.
\end{enumerate}
\end{mdframed}
    \caption{Low-weight linearity test.}
    \label{fig:(LW)LinearityTest}
\end{figure}

For this game we let $\mu$ denote the implied probability distribution on possible questions $\I_A \times \I_B$. We see that for any $x\in \I_A$ and $y\in\I_B$, either $\pi(x,y) =0$ or $\frac{1}{\pi(x,y)} = O(n^6)$.

Suppose $S=(\tau^A,\tau^B,\ket{\psi}\in\mcH_A\otimes\mcH_B)$ is a quantum strategy for this game. On each input $x\in \I_A$, Alice performs a two-outcome PVM $\{E^x_0,E^x_1\}$ where $E^x_0-E^x_1=\tau^A(x)$. Similarly, on each input
$y=(W(a) , W(a') ) \in I_B$ Bob performs a four-outcome PVM $\{ F^y_{0,0}, F^y_{1,0}, F^y_{0,1} , F^y_{1,1} \}$, and we define the following observables:
\begin{align*}
    &\tau^B_y(W(a)):=F^y_{0,0}+F^y_{0,1}-F^y_{1,0}-F^y_{1,1},\\
    &\tau^B_y(W(a')):=F^y_{0,0}+F^y_{1,0}-F^y_{0,1}-F^y_{1,1},\text{ and }\\
    &\tau^B_y(W(a+a')):=F^y_{0,0}+F^y_{1,1}-F^y_{0,1}-F^y_{1,0}.
\end{align*}
We take a moment to explain some subtleties of the above notation. Instead of specifying a single observable $\tau^B(W(a))$ for all possible $W(a)$ we include a reference to each particular question pair $y=(W(a),W(a'))$. This is because it is possible for the same $W(a)$ to be included in different question pairs for different choices of $W(a')$. In particular, we can not assume that Bob's measurements will be the same for $W(a)$ across these different possible choices. For Alice, things are more straightforward and we can define a unique observable $\tau^A(W(a))$ for each $W(a)$. Note that one can readily check that for each $y=(W(a),W(a'))$ we have
\begin{align}
    \tau^B_y(W(a+a'))=\tau^B_y(W(a))\tau^B_y(W(a'))=\tau^B_y(W(a'))\tau^B_y(W(a)).
\end{align}

\begin{lemma}\label{Lem:LWLT}
Suppose $S=(\tau^A,\tau^B,\ket{\psi}\in\H_A\otimes\H_B)$ is a strategy for low-weight linearity test with winning probability $1-\epsilon$. Let $\rho_A:=\Tr_{\H_B}(\ket{\psi}\bra{\psi})$. Then the following bounds hold.
\begin{enumerate}
    \item For every $y=(W(a),W(a'))\in\I_B$,
    \begin{enumerate}
        \item $\bra{\psi}\tau^A(W(a))\otimes\tau^B_y(W(a))\ket{\psi}\geq 1-O(n^6)\epsilon$,
        \item $\bra{\psi}\tau^A(W(a'))\otimes\tau^B_y(W(a'))\ket{\psi}\geq 1-O(n^6)\epsilon$, and
        \item if in addition, $W(a+a')\in\I_A$, then $\bra{\psi} \tau^A(W(a+a'))\otimes \tau^B_y(W(a+a'))\geq 1-O(n^6)\epsilon$.
    \end{enumerate}

    \item\label{ineqA} For every $W \in \{X,Z\}^{n}$ and $a,a' \in \lbrace 0,1 \rbrace^{n} $ with $\abs{a},\abs{a'},\abs{a+a'}\leq 6$,
    \begin{enumerate}
        \item $\norm{\tau^A(W(a))\tau^A(W(a'))-\tau^A(W(a+a'))}^2_{\rho_A}\leq O(n^6)\epsilon$,
        \item $\norm{\tau^A(W(a'))\tau^A(W(a))-\tau^A(W(a+a'))}^2_{\rho_A}\leq O(n^6)\epsilon $, and
        \item $\norm{\tau^A(W(a))\tau^A(W(a'))-\tau^A(W(a'))\tau^A(W(a))}^2_{\rho_A}\leq O(n^6)\epsilon $
    \end{enumerate}
\end{enumerate}
\end{lemma}
\begin{proof}
Suppose $S=(\tau^A,\tau^B,\ket{\psi})$ wins this game with probability $1- \epsilon$.
If we condition on a round in which  Bob receives question $ y= (W(a), W(a'))$ and Alice receives $x = W(a)$ then the probability they lose on this question pair is bounded above by $(1/ \pi(x,y)) \epsilon\leq O(n^6)\epsilon$. Indeed otherwise they will lose the overall game with a probability greater than $\epsilon$. Hence for any $x=W(a)\in\I_A$ and $y=(W(a),W(a'))\in\I_B$, we have
\begin{align*}
    \bra{\psi} \tau^A(W(a))\otimes \tau^B_y(W(a)) \ket{\psi} =& \bra{\psi}(E^x_0-E^x_1)\otimes (F^y_{0,0}+F^y_{0,1}-F^y_{1,0}-F^y_{1,1})\ket{\psi}\\
    =&\bra{\psi}(E^x_0\otimes F^y_{0,0}+E^x_0\otimes F^y_{0,1}+E^x_1\otimes F^y_{1,0}+E^x_1\otimes F^{y}_{1,1})\ket{\psi}\\
    &-\bra{\psi}(E^x_0\otimes F^y_{1,0}+E^x_0\otimes F^y_{1,1}+E^x_1\otimes F^y_{0,0}+E^x_1\otimes F^{y}_{0,1})\ket{\psi}\\
    =&Pr(\text{win}|x,y)-Pr(\text{lose}|x,y)\\
    =&1-2Pr(\text{lose}|x,y)\geq 1-O(n^6)\epsilon
\end{align*}
A similar analysis holds in the case that Bob receives question $ y= (W(a), W(a'))$ and Alice instead receives question $x = W(a')$. In particular we get
\begin{align*}
    \bra{\psi} \tau^A(W(a'))\otimes \tau^B_y(W(a')) \ket{\psi}\geq 1-O(n^6)\epsilon.
\end{align*}
Next we consider the instances in which Bob receives question $ y= (W(a), W(a'))$ and Alice receives question $x= W(a+a')$ where $\abs{a+a'}\leq 6$. In this case the winning condition requires $b_1 +b_2 = c$, so we have
\begin{align*}
    \bra{\psi} \tau^A(W(a+a'))\otimes \tau^B_y(W(a+a'))=& \bra{\psi}(E^x_0-E^x_1)\otimes (F^y_{0,0}+F^y_{1,1}-F^y_{0,1}-F^y_{1,0})\ket{\psi}\\
    =&\bra{\psi}(E^x_0\otimes F^y_{0,0}+E^x_0\otimes F^y_{1,1}+E^x_1\otimes F^y_{1,0}+E^x_1\otimes F^{y}_{0,1})\ket{\psi}\\
    &-\bra{\psi}(E^x_0\otimes F^y_{1,0}+E^x_0\otimes F^y_{0,1}+E^x_1\otimes F^y_{0,0}+E^x_1\otimes F^{y}_{1,1})\ket{\psi}\\
    =&Pr(\text{win}|x,y)-Pr(\text{lose}|x,y)\\
    =&1-2Pr(\text{lose}|x,y)\geq 1-O(n^6)\epsilon.
\end{align*}
To see the inequalities in the second part hold for fixed $W(a)$ and $W(a')$ where $\abs{a+a'}\leq 6$, we take $x=W(a+a'),y=(W(a),W(a'))$ to be the corresponding inputs for Alice and Bob. Let $\rho:=\ket{\psi}\bra{\psi}$. By the above work, we have
\begin{align*}
    &\norm{\tau^A(W(a))\otimes\Id-\Id\otimes\tau^B_y(W(a))}_\rho^2=2-2\bra{\psi}\tau^A(W(a))\otimes\tau^B_y(W(a))\ket{\psi}\leq O(n^6)\epsilon,\\
    &\norm{\tau^A(W(a'))\otimes\Id-\Id\otimes\tau^B_y(W(a'))}_\rho^2=2-2\bra{\psi}\tau^A(W(a'))\otimes\tau^B_y(W(a'))\ket{\psi}\leq O(n^6)\epsilon,\text{ and }\\
    &\norm{\tau^A(W(a+a'))\otimes\Id-\Id\otimes\tau^B_y(W(a+a'))}_\rho^2=2-2\bra{\psi}\tau^A(W(a+a'))\otimes\tau^B_y(W(a+a'))\ket{\psi}\leq O(n^6)\epsilon.
\end{align*}
Note that $\tau^B_y(W(a+a'))=\tau^B_y(W(a'))\tau^B_y(W(a))$ and that $\tau^A(W(a))$ and $\tau^B_y(W(a'))$ are unitary operators. We have
\begin{align*}
   \norm{\tau^A(W(a))\tau^A(W(a'))-\tau^A(W(a+a'))}_{\rho_A}
   =& \norm{\Big(\tau^A(W(a))\tau^A(W(a'))-\tau^A(W(a+a'))\Big)\otimes\Id}_\rho\\
   \leq&\norm{\tau^A(W(a))\tau^A(W(a'))\otimes\Id-\tau^A(W(a))\otimes\tau^B_y(W(a'))}_\rho\\
   +&\norm{\tau^A(W(a))\otimes\tau^B_y(W(a'))-\Id\otimes \tau^B_y(W(a'))\tau^B_y(W(a))}_\rho\\
   +&\norm{\Id\otimes \tau^B_y(W(a+a'))-\tau^A(W(a+a'))\otimes \Id}_\rho\\
    \leq&O(n^3)\sqrt{\epsilon}.
\end{align*}
Hence 2(a) holds, 2(b) holds for the same reason, while 2(c) follows from 2(a) and 2(b).
\end{proof}

\subsection{Low-weight anti-commutation test}
\label{sec:LW-anti-commutation}

  Next we introduce a natural version of the anti-commutation test built from the well-known Magic Square game which we described in \Cref{subsec:rigidity}. To incorporate this game in our anti-commutation test, we ask the players to play the Magic Square game on a specified 2-qubit register of their shared $n$-qubit state. Importantly, for all but one question $v_9$ Alice will not be able to tell if she is playing the low-weight linearity test or anti-commutation test, and thus we can combine the rigidity constraints from each game. Recall that the system of equations in the Magic Square game has an operator solution $\{A_1,\cdots, A_9\}$ defined in \Cref{Alice}. We describe the low-weight anti-commutation test in \Cref{fig:(LW)Anti-commutationTest}.
\begin{figure}[h]
    \centering
    \begin{mdframed}
    \begin{enumerate}
    \item The verifier samples uniformly at random a string $a \in \{ 0, 1 \}^{n}$ with exactly two non-zero entries $i<j$. The verifier also samples a row or column $q \in \lbrace r_1, r_2, r_3, c_1, c_2,c_3 \rbrace$, and a variable $v_k$ contained in $q$ as in the Magic Square game.
    \item Bob receives the question $(q,a)$.
    \item If $k \neq 9$ then Alice receives $W(a)=I^{i-1}W_iI^{j-i}W_jI^{n-j}\in \mcI_A$ with $\sigma_{W_i} \otimes \sigma_{W_j}=A_k$. If $k=9$ then Alice receives question $(v_9, a)$.
    \item The players win if and only if Bob responds with a satisfying assignment to $q$ and  Alice provides an assignment to variable $v_k$ that is consistent with Bob's.
\end{enumerate}
\end{mdframed}
    \caption{Low-weight anti-commutation test.}
    \label{fig:(LW)Anti-commutationTest}
\end{figure}

In the above game, for each string $a \in \lbrace 0 ,1 \rbrace^{n}$ with $|a|=2$ we define the following set of observables for Alice $\lbrace O_1(a), \dots, O_9(a) \rbrace.$ If the players are winning the anti-commutation test with probability $1-\epsilon$ then the above observables must determine the operators of a quantum strategy for Magic Square that wins with probability at least $1-O(n^2)\epsilon$. Applying the rigidity of the Magic Square game (\Cref{lemma:MS}) we get the following Lemma.
\begin{lemma}\label{lem:anticommutation}
If the players are winning low-weight anti-commutation test with probability $1-\epsilon$, then for any $a \in \lbrace 0,1 \rbrace^{n}$ with $|a|=2$ we have
\begin{align*}
    &\norm{O_{1}(a) O_{5}(a) + O_{5}(a)O_{1}(a)}^2_{\sigma}\leq O(n){\sqrt{\epsilon}}, \text{ and }\\
    &\norm{O_{4}(a) O_{2}(a) + O_{2}(a)O_{4}(a)}^2_{\sigma} \leq O(n){\sqrt{\epsilon}}.
\end{align*}
\end{lemma}

\subsection{Formal statements}
\label{sec:LW-PBT}

Combining the low-weight linearity test and low-weight anti-commutation test from the previous sections, we now construct the low-weight Pauli braiding test and state its rigidity result.

\begin{defn}
The low-weight Pauli braiding test (LWPBT) is played by executing with probability $1/2$ either the low-weight anti-commutation test or the low-weight linearity test.
\end{defn}

If a quantum strategy $S=(\tau^A,\tau^B,\ket{\psi})$ can win LWPBT with probability $1-\epsilon$, then it must be winning both the low-weight linearity test and low-weight anti-commutation test with probability at least $1-2\epsilon$. For every $W\in\{X,Z\}$ and $i\in[n]$, we use $\tau_W^A(e_i)$ to denote $\tau^A(I^{i-1}WI^{n-i})$. For any bit string $a=0^{j-2}110^{n-j}$ with $2\leq j\leq n$, since Alice cannot tell a question $I^{j-1}WI^{n-j}$ is from the low-weight linearity test or low-weight anti-commutation test, we must have $O_1(a)=\tau_X^A(e_j)$ and $O_5(a)=\tau_Z^A(e_j)$. Similarly, for any bit string $a=0^{j-1}110^{n-j-1}$ with $1\leq j\leq n-1$ we must have $O_2(a)=\tau_X^A(e_j)$ and $O_4(a)=\tau_Z^A(e_j)$. Then the following theorem follows directly from \Cref{Lem:LWLT} and \Cref{lem:anticommutation}.

\begin{thm}\label{polyn}
Suppose $S=(\tau^A,\tau^B,\ket{\psi}\in\H_A\otimes\H_B)$ is a strategy for LWPBT with wining probability $1-\epsilon$. Let $\rho_A:=\Tr_{\H_B}(\ket{\psi}\bra{\psi})$.
\begin{enumerate}
    \item For every $y=(W(a),W(a'))\in\I_B$, we have the following bounds on consistency.
    \begin{enumerate}
        \item $\bra{\psi}\tau^A(W(a))\otimes\tau^B_y(W(a))\ket{\psi}\geq 1-O(n^6)\epsilon$,
        \item $\bra{\psi}\tau^A(W(a'))\otimes\tau^B_y(W(a'))\ket{\psi}\geq 1-O(n^6)\epsilon$, and
        \item if in addition $\abs{a+a'}\leq 6$, then $\bra{\psi} \tau^A(W(a+a'))\otimes \tau^B_y(W(a+a'))\geq 1-O(n^6)\epsilon$.
    \end{enumerate}

    \item For every $W \in \{X,Z\}^{n}$ and  $a,a' \in \{ 0,1 \}^{n} $ with $\abs{a},\abs{a'},\abs{a+a'}\leq 6$, we have the following linearity bounds and commutation bound.
    \begin{enumerate}
        \item $\norm{\tau^A(W(a))\tau^A(W(a'))-\tau^A(W(a+a'))}^2_{\rho_A}\leq O(n^6)\epsilon$,
        \item $\norm{\tau^A(W(a'))\tau^A(W(a))-\tau^A(W(a+a'))}^2_{\rho_A}\leq O(n^6)\epsilon $, and
        \item $\norm{\tau^A(W(a))\tau^A(W(a'))-\tau^A(W(a'))\tau^A(W(a))}^2_{\rho_A}\leq O(n^6)\epsilon $
    \end{enumerate}

    \item For any $i, j\in [n]$ we have the following commutation/anti-commutation bound.
    \begin{enumerate}
        \item $\norm{\tau^A_X(e_i)\tau^A_Z(e_j)-(-1)^{\delta_{ij}}\tau^A_Z(e_j)\tau^A_X(e_i)}^2_{\rho_A}\leq O(n)\sqrt{\epsilon}$.
    \end{enumerate}
\end{enumerate}
\end{thm}

%% file: Rigidity.tex
In this section, we use a group-theoretical approach to analyze the LWPBT. We show that every near-optimal strategy for LWPBT forms an approximate homomorphism of the Weyl-Heisenberg group. For rounding approximate homomorphisms to exact representations, we provide an enhanced Gowers-Hatami theorem. Applying this new stability result to the Weyl-Heisenberg group, we prove that any strategy for LWPBT with winning probability $1-\epsilon$ must be $poly(n)\sqrt{\epsilon}$-close to the canonical perfect strategy.

\subsection{An enhanced Gowers-Hatami theorem}\label{GH}

The Gowers-Hatami theorem \cite{GH17} and its variant \cite{Vid18} play an important role in the rigidity analysis for nonlocal games.
The theorem states that every approximate homomorphism of a finite group is close to a representation. Some nonlocal games $\mcG$ can be modeled by a finite group~$G$ in the sense that optimal strategies for $\mcG$ correspond to representations of $G$ and near-optimal strategies for $\mcG$ correspond to approximate homomorphisms of~$G$. If in addition, $\mcG$ has a canonical optimal strategy $\wtd{\mcS}$, then this theorem implies every near-optimal strategy $\mcS$ for $\mcG$ is close to~$\wtd{\mcS}$. However, as discussed in \Cref{sec:intro} and \Cref{subsec:approxrep}, some subtle mathematical problems have come up in earlier approaches. In particular, one may need to discard some irreducible components of a representation that do not correspond to $\wtd{\mcS}$. In the following enhanced Gowers-Hatami theorem, we observe that as long as an approximate homomorphism satisfies some additional symmetry, it must be close to a representation containing no such ``junk" irreducible components.

\begin{thm}\label{thm:GH}
    Let $G$ be a finite group. Let $f:G\arr \msU(\mcH)$ be an $(\epsilon,\rho)$-homomorphism, and suppose~$S$ is a subgroup of $G$ satisfying
\begin{align}
    f(sg)=f(s)f(g) \text{ for all } s\in S, g\in G.
\end{align}
Then $f|_S:S\arr\msU(\mcH)$ is a representation of $S$, and there exists a Hilbert space $\mcK$, an isometry $V:\mcH\arr\mcK$, and a representation $\phi:G\arr\msU(\mcK)$ such that
    \begin{align}
        \norm{Vf(g)-\phi(g)V}^2_\rho\leq \epsilon \text{ for all } g\in G,
    \end{align}
and every irreducible component $\xi$ of $\phi$ satisfies $\hat{f|_S}(\xi|_S)\neq 0$.
\end{thm}
\begin{proof}
     $f|_S$ respects the group structure of $S$, so it is a representation of $S$. Let $\Irr(G,f|_S):=\{\varphi\in \Irr(G): \hat{f|_S}(\varphi|_S)\neq 0 \}$. For any $\varphi\in \Irr(G)$ and $s\in S$, we have
     \begin{align*}
       \big(f(s)\otimes\overline{\phi(s)}\big)\hat{f}(\varphi)=\frac{1}{\abs{G}}\sum_{g\in G}f(s)f(g)\otimes \overline{\varphi(s)\varphi(g)}=\frac{1}{\abs{G}}\sum_{g\in G}f(sg)\otimes \overline{\varphi(sg)}=\hat{f}(\varphi).
     \end{align*}
     Since $\hat{f|_S}(\varphi|_S)=\frac{1}{\abs{S}}\sum_{s\in S}\big(f(s)\otimes\overline{\phi(s)}\big)\hat{f}(\varphi)$, it follows that
    \begin{align}
        \hat{f|_S}(\varphi|_S)\cdotp\hat{f}(\varphi)=\hat{f}(\varphi)\label{fourier}.
    \end{align}
In particular, for any $\xi\in \Irr(G)\setminus \Irr(G,f|_S)$, since $\hat{f|_S}(\xi|_S)=0$, \Cref{fourier} implies $\hat{f}(\xi)=0$. Let $\mcK:=\bigoplus\limits_{\varphi\in \Irr(G,f|_S)}\C^{d_\varphi}\otimes \mcH\otimes\C^{d_\varphi}$, and define a linear map $V:\mcH\arr\mcK$ via sending
\begin{align*}
    \ket{v}\mapsto \bigoplus_{\varphi\in \Irr(G,f|_S)}\sqrt{d_\varphi}\sum_{i=1}^{d_\varphi}\ket{i}\otimes\Big(\hat{f}(\varphi)(\ket{v}\otimes\ket{i})\Big),
\end{align*}
where $\{\ket{i}\}_{i=1}^{d_\varphi}$ is the standard basis for $\C^{d_\varphi}$. Since $\hat{f}(\varphi)=0$ for all $\varphi\in Irr(G)\setminus Irr(G,f|_S)$,
\begin{align*}
    V^*V&=\sum_{\varphi\in \Irr(G,f|_S)}d_\varphi\sum_{i=1}^{d_\varphi}\big(\Id_{\mcH}\otimes\bra{i}\big)\hat{f}(\varphi)^*\hat{f}(\varphi)\big(\Id_{\mcH}\otimes\ket{i}\big)\\
    &=\sum_{\varphi\in \Irr(G)}d_\varphi\sum_{i=1}^{d_\varphi}\big(\Id_{\mcH}\otimes\bra{i}\big)\hat{f}(\varphi)^*\hat{f}(\varphi)\big(\Id_{\mcH}\otimes\ket{i}\big)\\
    &=\frac{1}{\abs{G}^2}\sum_{\varphi\in \Irr(G)}d_{\varphi}\sum_{g,h\in G}f(g)^*f(h)\sum_{i=1}^{d_\varphi}\bra{i}\varphi(g)^T\overline{\varphi(h)}\ket{i}\\
    &=\frac{1}{\abs{G}^2}\sum_{g,h\in G}f(g)^*f(h)\sum_{\varphi\in \Irr(G)}d_{\varphi}\Tr\big(\varphi(g^{-1}h)  \big)\\
    &=\frac{1}{\abs{G}}\sum_{g\in G}f(g)^*f(g)=\Id_{\mcH}.
\end{align*}
Therefore, $V$ is an isometry. Now consider the representation $\phi:G\arr\msU(\mcK)$ defined by sending
\begin{align*}
    g\mapsto\bigoplus_{\varphi\in \Irr(G,f|_S)} \varphi(g)\otimes \Id_{\mcH}\otimes \Id_{\C^{d_\varphi}}.
\end{align*}
We see that $\varphi$ is an irreducible component of $\phi$ only if $\varphi\in \Irr(G,f|_S)$. Note that for any $A\in M_d(\C)$ we have $\bra{i}A\ket{j}=\bra{j}A^T\ket{i}$ for all basis vectors $\ket{i}$ and $\ket{j}$. Then for any $g\in G$, we obtain that
\begin{align*}
    V^*\phi(g)V&=\sum_{\varphi\in \Irr(G,f|_S)}d_\varphi\sum_{i,j=1}^{d_\varphi}\bra{i}\varphi(g)\ket{j}\big(\Id_{\mcH}\otimes\bra{i}\big)\hat{f}(\varphi)^*\hat{f}(\varphi)\big(\Id_{\mcH}\otimes\ket{j}\big)\\
    &=\sum_{\varphi\in \Irr(G)}d_\varphi\sum_{i,j=1}^{d_\varphi}\bra{i}\varphi(g)\ket{j}\big(\Id_{\mcH}\otimes\bra{i}\big)\hat{f}(\varphi)^*\hat{f}(\varphi)\big(\Id_{\mcH}\otimes\ket{j}\big)\\
    &=\frac{1}{\abs{G}^2}\sum_{\varphi\in \Irr(G)}d_\varphi\sum_{h,k\in G}f(h)^*f(k)\sum_{i,j=1}^{d_\varphi}\bra{i}\varphi(g)\ket{j}\bra{i}\varphi(h)^T\overline{\varphi(k)}\ket{j}\\
    &=\frac{1}{\abs{G}^2}\sum_{\varphi\in \Irr(G)}d_\varphi\sum_{h,k\in G}f(h)^*f(k)\sum_{i,j=1}^{d_\varphi}\bra{i}\varphi(g)\ket{j}\bra{j}\varphi(k)^*\varphi(h)\ket{i}\\
 &=\frac{1}{\abs{G}^2}\sum_{h,k\in G}f(h)^*f(k)\sum_{\varphi\in \Irr(G)}d_\varphi \Tr\big(\varphi(k^{-1}hg)  \big)=\frac{1}{\abs{G}}\sum_{h\in G}f(h)^*f(hg).
\end{align*}
Then we get
\begin{align*}
    \norm{Vf(g)-\phi(g)V}_\rho^2&=2-\frac{1}{\abs{G}}\sum_{h\in G}2\mathfrak{Re}\Tr\big(f(g)^*f(h)^*f(hg)\rho\big)\\&=\frac{1}{\abs{G}}\sum_{h\in G}\Big(2-2\mathfrak{Re}\Tr\big(f(g)^*f(h)^*f(hg)\rho\big)  \Big)\\
   &=\frac{1}{\abs{G}}\sum_{h\in G}\norm{f(h)f(g)-f(hg)}^2_\rho.
\end{align*}
Since $f$ is an $(\epsilon,\rho)$-homomorphism of $G$, it follows that $\norm{Vf(g)-\phi(g)V}_\rho^2\leq\epsilon$ for all $g\in G$.
\end{proof}

\subsection{Proof of rigidity}\label{subsec:ProofOfRigidity}

We now aim to prove the rigidity of the low-weight Pauli Braiding test. To start, we fix an $n\in\N$, an $\epsilon>0$, and a quantum strategy  $S:=(\tau^A,\tau^B,\ket{\psi}\in\H_A\otimes\H_B)$ for the $n$-qubit LWPBT with winning probability $1-\epsilon$ through this section (except for \Cref{rigidity}).

Let $H(n)$ be the $n$-qubit Weyl-Heisenberg group generated by indeterminates $\{J,X_i,Z_i,i\in[n]\}$ subject to the relations
\begin{enumerate}
    \item[(R0)] $J$ is central, and $J^2=X_i^2=Z_i^2=1$ for all $1\leq i\leq n$,
    \item[(R1)] $[X_i,Z_i]=J$ for all or all $1\leq i\leq n$, and
    \item[(R2)] $[X_i,X_j]=[Z_i,Z_j]=[X_i,Z_j]=1$ for all $i\neq j$.
\end{enumerate}

Intuitively, this group models $n$-qubit Pauli measurements. It has an irreducible representation $\sigma:H(n)\arr\msU((\C^2)^{\otimes n})$ sending $J\mapsto -\Id$, $X_i\mapsto \sigma_X(e_i)$, and $Z_i\mapsto \sigma_Z(e_i)$ for all $i\in [n]$. In this section, the symbol $\sigma$ will be used exclusively to refer to this irreducible representation. The following two lemmas are well-known.

\begin{lemma}\label{lemma:irreps}
$\Irr(H(n))=\{\sigma\}\cup\{\chi_i:1\leq i\leq 2^{2n}\}$ where every $\chi_i:H(n)\arr\C$ is a one dimensional representation sending $J\mapsto 1$.
\end{lemma}

\begin{lemma}$H(n)$ has a normal form
   $ \mathcal{N}(w)=J^\alpha X_{i_1}\cdots X_{i_k}Z_{j_1}\cdots Z_{j_l} $
where $\alpha\in \{0,1\}$, $1\leq i_1<\cdots <i_k\leq n$, and $1\leq j_1<\cdots<j_l\leq n$.
\end{lemma}

For any $W\in\{X,Z\}$ and $i\in [n]$, we use $\tau^A_W(e_i) $ and $ \tau^B_W(e_i)$ to denote $\tau^A(I^{i-1}WI^{n-i})$ and $\tau^B_y(I^{i-1}WI^{n-i})$ respectively, where $y=(I^{i-1}WI^{n-i},I^n)$. Define functions $f_P:H(n)\arr\msU(\mcH_P)$, $P\in\{A,B\}$ via sending
\begin{align*}
   J^\alpha X_{i_1}\cdots X_{i_k}Z_{j_1}\cdots Z_{j_l}\mapsto (-1)^\alpha\tau_X^P(e_{i_1})\cdots\tau_X^P(e_{i_k})\tau_Z^P(e_{j_1})\cdots\tau_Z^P(e_{j_l}),
\end{align*}
so in particular, $f_P(1)=\Id_{\mcH_P}$ and $f_P(J)=-\Id_{\mcH_P}$. We also define functions $\widetilde{f}_A,\widetilde{f}_B:H(n)\arr\msU(\mcH_A\otimes\mcH_B)$ via $\widetilde{f}_A(g):=f_A(g)\otimes \Id_{\H_B}$ and $\widetilde{f}_B(g):=\Id_{\H_A}\otimes f_B(g)$. Through this section, we use $\rho,\rho_A$ and $\rho_B$ to denote $\ket{\psi}\bra{\psi},\Tr_{\H_B}(\rho)$ and $\Tr_{\H_A}(\rho)$ respectively.

\begin{prop}\label{prop:4conditions}
    Let $\delta:=O(n^6)\sqrt{\epsilon}$. Every $f\in\{\widetilde{f}_A,\widetilde{f}_B\}$ satisfies the following conditions.
    \begin{enumerate}
    \item[(1)] $f$ respects all the relations in R0.
    \item[(2)] $\norm{f(X_i)f(Z_i)+f(Z_i)f(X_i)}^2_{\rho}\leq\delta$ for all $1\leq i\leq n$.
    \item[(3)]  $\norm{f(X_i)f(X_j)-f(X_j)f(X_i)}^2_{\rho},\norm{f(Z_i)f(Z_j)-f(Z_j)f(Z_i)}^2_{\rho}$, and $\norm{f(X_i)f(Z_j)-f(Z_j)f(X_i)}^2_{\rho}$ are bounded by $\delta$ whenever $i\neq j$.
    \item[(4)] For every $W_i$ where $W\in\{X,Z\},i\in [n]$, there exists a unitary $\widehat{W}_i$ commuting with  all $\{f(X_j),f(Z_j),j\in [n]\}$ such that $\norm{f(W_i)-\widehat{W}_i}^2_{\rho}\leq \delta$.
    \item[(5)] For any $s\geq 1$ and every monomial\footnote{Here we allow the coefficient of a monomial to be 1 or $-1$.} $v=v_1\cdots v_s$ over $\{f(X_i),f(Z_i),i\in[n]\}$ of degree $s$, there exists a unitary $\hat{v}$ commuting with all $\{f(X_i),f(Z_i),i\in [n]\}$ such that $\norm{v-\hat{v}}_\rho\leq s\sqrt{\delta}$.
\end{enumerate}
\end{prop}

\begin{proof}
    We first examine that $\widetilde{f}_A$ satisfies (1), (2), (3), and (4). Since $\norm{W\otimes \Id_{\H_B}}_\rho=\norm{W}_{\rho_A}$ for all $W\in\mathcal{B}(\H_A)$, it follows from \Cref{polyn} that $\widetilde{f}_A$ satisfies part (1), (2), and (3). To establish part (4), observe that every $\widetilde{f}_B(W_i)$ commutes with $\widetilde{f}_A(X_j),\widetilde{f}_A(Z_j)$ for all $j\in[n]$, and
\begin{align}
    \norm{\widetilde{f}_A(W_i)-\widetilde{f}_B(W_i)}^2_\rho=2-2\bra{\psi}\tau_W^A(e_i)\otimes\tau_W^B(e_i)\ket{\psi}\leq O(n^6)\sqrt{\epsilon}\label{ABdiff}
\end{align}
by \Cref{polyn}. Hence part (4) follows by taking $\hat{W}_i=\widetilde{f}_B(W_i)$.

Next we examine that $\widetilde{f}_B$ satisfies (1), (2), (3) and (4). Part (1) is obvious. Since
\begin{align*} \widetilde{f}_B(X_i)\widetilde{f}_B(Z_i)+\widetilde{f}_B(Z_i)\widetilde{f}_B(X_i)&=\widetilde{f}_B(X_i)\big(\widetilde{f}_B(Z_i)-\widetilde{f}_A(Z_i) \big)+\widetilde{f}_B(Z_i)\big(\widetilde{f}_B(X_i)-\widetilde{f}_A(X_i)\big)\\
    &+ \widetilde{f}_A(Z_i)\big(\widetilde{f}_B(X_i)-\widetilde{f}_A(X_i) \big)+\widetilde{f}_A(X_i)\big(\widetilde{f}_B(Z_i)-f_A(Z_i) \big)\\
    &+\widetilde{f}_A(Z_i)\widetilde{f}_A(X_i)+\widetilde{f}_A(X_i)\widetilde{f}_A(Z_i),
\end{align*}
and each term on the right hand side has $\rho$-norm $\leq O(n^6)\sqrt{\epsilon}$ by \Cref{polyn}, it follows that $\norm{\widetilde{f}_B(X_i)\widetilde{f}_B(Z_i)+\widetilde{f}_B(Z_i)\widetilde{f}_B(X_i)}^2_\rho\leq O(n^6)\sqrt{\epsilon}$. Therefore, part (2) holds for $\widetilde{f}_B$ . With a similar argument we obtain that part (3) holds. Part (4) follows from \Cref{ABdiff} by taking $\hat{W}_i=\widetilde{f}_A(W_i)$.

Now we prove part (5) by induction on the monomial degree $s\geq 1$. The base case $s=1$ follows straight from part (4). Suppose part (5) holds for all monomials over $\{f(X_i),f(Z_i),i\in[n]\}$ of degree $s$. For any monomial $v=v_1\cdots v_sv_{s+1}$ over $\{f(X_i),f(Z_i),i\in[n]\}$ of degree $s+1$, let $v':=v_1\cdots v_s$. By the induction hypothesis, there are unitaries $\hat{v}'$ and $\hat{v}_{s+1}$ commuting with all $\{f(X_i),f(Z_i),i\in[n]\}$ such that $\norm{v'-\hat{v}'}_\rho\leq s\sqrt{\delta}$ and $\norm{v_{s+1}-\hat{v}_{s+1}}_\rho\leq\sqrt{\delta}$. In particular, $\hat{v}_{s+1}$ commutes with $v'$, so $v=v'(v_{s+1}-\hat{v}_{s+1})+\hat{v}_{s+1}(v'-\hat{v}')+\hat{v}_{s+1}\hat{v}'$. Since $\rho$-norm is left unitarily invariant and $\hat{v}',\hat{v}_{s+1}$ are unitaries commuting with all $\{f(X_i),f(Z_i),i\in[n]\}$, it follows that $\hat{v}_{s+1}\hat{v}'=:\hat{v}$ is a unitary commuting with all all $\{f(X_i),f(Z_i),i\in[n]\}$ and satisfies
\begin{align*}
    \norm{v-\hat{v}}_\rho\leq \norm{v_{s+1}-\hat{v}_{s+1}}_\rho+\norm{v'-\hat{v}'}_\rho\leq (s+1)\sqrt{\delta}.
\end{align*}
We conclude that part (5) holds for all monomials over $\{f(X_i),f(Z_i),i\in[n]\}$.
\end{proof}

\begin{prop}
    Let $\delta:=O(n^6)\sqrt{\epsilon}$. Then every $f\in\{\widetilde{f}_A,\widetilde{f}_B\}$ is an $\big(O(n^6)\delta,\rho\big)$-homomorphism of $H(n)$.
\end{prop}

\begin{proof}
Let $R:=\{f(X_i)f(X_j)-f(X_j)f(X_i),f(Z_i)f(Z_j)-f(Z_j)f(Z_i),f(X_i)f(Z_j)-f(Z_j)f(X_i):i\neq j\}\cup \{f(X_i)f(Z_i)+f(Z_i)f(X_i):i\in [n]\}$. Then by \Cref{prop:4conditions} part (2) and (3), $\norm{r}_\rho\leq \sqrt{\delta}$ for all $r\in R$. Since $f(X_i)$ and $f(Z_i)$ are unitaries, it follows that $\norm{r}_{op}\leq 2$ for all $r\in R$.

By the definition of $\wtd{f}_A$ and $\wtd{f}_B$ and the normal form theorem for $H(n)$, for any $g,h\in H(n)$, we see that $f(g)f(h)$ and $f(gh)$ are monomials over  $\{f(X_i),f(Z_i),i\in[n]\}$ of degrees less than $4n$, and $f(g)f(h)-f(gh)=\sum_{i=1}^ku_ir_iv_i$ where $k\leq 4n^2$, $r_1,\cdots r_k\in R$, $u_1,\cdots u_k,v_1,\cdots, v_k$ are monomials over $\{f(X_i),f(Z_i),i\in[n]\}$, and every $v_i$ has degree less than $2n$. The last condition implies that for every $i\in [k]$, there is a unitary $\hat{v}_i$ commuting with all $\{f(X_i),f(Z_i),i\in[n]\}$ such that
$\norm{v_i-\hat{v}_i}_\rho\leq 2n\sqrt{\delta}$ by \Cref{prop:4conditions} part (5). Then for each $i\in [k]$ we have
\begin{align*}
    \norm{u_ir_iv_i}_\rho&=\norm{r_i(v_i-\hat{v}_i+\hat{v_i})}_\rho\leq \norm{r_i(v_i-\hat{v}_i)}_\rho+\norm{\hat{v_i}r_i}_\rho\\
    &\leq \norm{r_i}_{op}\cdotp \norm{v_i-\hat{v}_i}_\rho+\norm{r_i}_\rho\leq (4n+1)\sqrt{\delta}.
\end{align*}
Hence $\norm{f(g)f(h)-f(gh)}_\rho\leq\sum_{i=1}^k \norm{u_ir_iv_i}_\rho\leq O(n^3)\sqrt{\delta}$. We conclude that $\norm{f(g)f(h)-f(gh)}_\rho^2\leq O(n^6)\delta$ for all $g,h\in H(n)$.
\end{proof}

Since $\norm{f_P(g)f_P(h)-f_P(gh)}^2_{\rho_P}=\norm{\widetilde{f}_P(g)\widetilde{f}_P(h)-\widetilde{f}_P(gh)}^2_\rho$ for every $P\in\{A,B\}$ and all $g,h\in H(n)$, we have:

\begin{cor}\label{cor:approxrep}
The function $f_A$ (resp. $f_B$) is an $(O(n^{12})\sqrt{\epsilon},\rho_A)$-homomorphism (resp. $(O(n^{12})\sqrt{\epsilon},\rho_B)$-homomorphism) of $H(n)$.
\end{cor}

Now we can apply \Cref{thm:GH} to $f\in\{f_A,f_B\}$. The subgroup $\ang{J}$ generated by $J$ contains only two elements $1$ and $J$, and we see that $f(sg)=f(s)f(g)$ for all $s\in\ang{J}$ and $g\in H(n)$. Hence $f|_{\ang{J}}$ is a representation of $H(n)$. In particular, $f$ sends $1\mapsto\Id$ and $J \mapsto -\Id$, so $\hat{f|_{\ang{J}}}(\chi|_{\ang{J}})=f(1)\chi(1)+f(J)\chi(J)=\Id-\Id=0$ for all one-dimensional irreducible representations $\chi$ of $H(n)$. This implies that $\sigma$ is the unique irreducible representation of $H(n)$ for which $\hat{f|_{\ang{J}}}(\sigma|_{\ang{J}})\neq 0$. Hence we have:

\begin{thm}\label{thm:isometries}
There is a Hilbert space $\mcH_A^{aux}$ (resp. $\mcH_B^{aux}$) and an isometry $V_A:\mcH_A\arr\C^{2^n}\otimes \mcH_A^{aux}$ (resp. $V_B:\mcH_B\arr\C^{2^n}\otimes \mcH_B^{aux}$) such that $\norm{V_Af_A(g)-\big(\sigma(g)\otimes \Id_{\mcK_A}\big)V_A}^2_{\rho_A}\leq O(n^{12})\sqrt{\epsilon}$ (resp. $\norm{V_Bf_B(g)-\big(\sigma(g)\otimes \Id_{\mcK_B}\big)V_B}^2_{\rho_B}\leq O(n^{12})\sqrt{\epsilon}$) for all $g\in H(n)$.
\end{thm}
\begin{proof}
    Let $\delta=O(n^{12})\sqrt{\epsilon}$. \Cref{cor:approxrep} implies $f_A$ is a $(\delta,\rho_A)$-homomorphism of $H(n)$. From the above discussion, we know that $\sigma$ is the unique irreducible representation for which $\hat{f_A|_{\ang{J}}}(\sigma|_{\ang{J}})\neq 0$. Then by \Cref{thm:GH}, there is a Hilbert space $\mcK$, an isometry $V:\mcH_A\arr\mcK$, and a representation $\phi:H(n)\arr\msU(\mcK)$ such that $\norm{Vf_A(g)-\phi(g)V}_\rho^2\leq\delta$ for all $g\in H(n)$, and $\phi$ has a unique irreducible component $\sigma$. The latter condition means that $\mcK\cong\C^{2^n}\otimes \mcH_A^{aux}$ for some Hilbert space $\mcH_A^{aux}$ and there is a unitary $U:\mcK\arr\C^{2^n}\otimes \mcH_A^{aux}$ such that $U\phi(g)U^*=\sigma(g)\otimes \Id_{\mcH_A^{aux}}$ for all $g\in H(n)$. Since $\rho_A$-norm is left unitarily invariant, it follows that
    \begin{align*}
        \norm{Vf(g)-\phi(g)V}_{\rho_A}&=\norm{U^*UVf(g)-U^*\big(\sigma(g)\otimes\Id_{\mcH_A^{aux}}\big)UV}_{\rho_A}\\
        &=\norm{UVf(g)-\big(\sigma(g)\otimes\Id_{\mcH_A^{aux}}\big)UV}_{\rho_A}
    \end{align*}
  for all $g\in H(n)$. We conclude that $V_A:=UV:\mcH_A\arr\C^{2^n}\otimes \mcH_A^{aux}$ is an isometry such that $\norm{V_Af_A(g)-\big(\sigma(g)\otimes Id_{\mcH_A^{aux}}\big)V_A}^2_{\rho_A}\leq \delta$. The argument for $f_B$ follows similarly.
\end{proof}

The above theorem illustrates that in a near-perfect strategy for LWPBT, the players must perform measurements that are close to the measurements in the canonical perfect strategy. The following lemma shows that the joint state shared by the players must also be close to a maximally entangled state.

\begin{lemma}\label{lemma:epr}
If $\ket{\psi}\in\C^{2^n}\otimes\C^{2^n}$ is a unit vector satisfying
\begin{align}
\frac{1}{4n}\sum_{a,b\in\{0,1\},i\in[n]}\bra{\psi}\big(\sigma_X(e_i)^a\sigma_Z(e_i)^b\otimes\sigma_X(e_i)^a\sigma_Z(e_i)^b \big)\ket{\psi}\geq 1-\delta\label{honestmeasurement}
\end{align}
for some $0\leq \delta\leq 1/n$, then $\abs{\braket{\psi}\eprn}^2\geq 1-n\delta$.
\end{lemma}

\begin{proof}
Let $S:=\frac{1}{4n}\sum\limits_{a,b\in\{0,1\},i\in[n]}\sigma_X(e_i)^a\sigma_Z(e_i)^b\otimes\sigma_X(e_i)^a\sigma_Z(e_i)^b$, and let $P_1:=\ket{\epr}\bra{\epr}$ and $P_0=\Id_{\C^4}-P_1$. Since
$\frac{1}{4}(\sigma_I\otimes\sigma_I+\sigma_X\otimes\sigma_X+\sigma_Z\otimes\sigma_Z+\sigma_X\sigma_Z\otimes\sigma_X\sigma_Z)=\ket{\epr}\bra{\epr}$, we have
\begin{align*}
    S=\frac{1}{n}\sum_{i\in[n]} \Id_{\C^4}^{\otimes (i-1)}\otimes \ket{\epr}\bra{\epr}\otimes \Id_{\C^4}^{\otimes (n-i)}
    =\frac{1}{n}\sum_{i\in[n]} (P_0+P_1)^{\otimes (i-1)}\otimes P_1\otimes (P_0+P_1)^{\otimes (n-i)}.
\end{align*}
This implies that $S$ is positive semi-definite and has a unique $+1$ eigenvector $\ket{\eprn}$. Any other eigenvector of $S$ has eigenvalue less than $\frac{n-1}{n}$. In other words, $S$ has a spectral decomposition $\sum_{i=1}^{2^{2n}}s_i\ket{\alpha_i}\bra{\alpha_i}$ where $s_1=1,s_i\in [0,\frac{n-1}{n}]$ for all $i\geq 2$, $\ket{\alpha_1}=\ket{\eprn}$ , and $\{\ket{\alpha_i}\}_{i=1}^{2^{2n}}$ is an orthonormal basis for $\C^{2^n}\otimes\C^{2^n}$. Now suppose $\ket{\psi}=\sum_{i=1}^{2^{2n}}\lambda_i\ket{\alpha_i}$ is a unit vector in $\C^{2^n}\otimes\C^{2^n}$ satisfying \Cref{honestmeasurement}. Then
\begin{align*}
    1-\delta
    &\leq \bra{\psi}S\ket{\psi}=\abs{\lambda_1}^2+\sum_{i\geq 2}s_i\abs{\lambda_i}^2\leq  \abs{\lambda_1}^2+\sum_{i\geq 2} (\frac{n-1}{n})\abs{\lambda_i}^2\\
    &=\abs{\lambda_1}^2+ (\frac{n-1}{n})(1-\abs{\lambda_1}^2)=\frac{n-1}{n}+\frac{1}{n}\abs{\lambda_1}^2.
\end{align*}
It follows that $\abs{\braket{\psi}\eprn}^2=\abs{\lambda_1}^2\geq 1-n\delta$.
\end{proof}

\begin{thm}\label{rigidity0}
There are Hilbert spaces $\mcH_A^{aux}$ and $\mcH_B^{aux}$, isometries $V_A:\H_A\rightarrow \C^{2n}\otimes\H_A^{aux},V_B:\H_B\rightarrow \C^{2n}\otimes\H_B^{aux}$, and a unit vector $\ket{aux}\in\H_A^{aux}\otimes\H_B^{aux}$ such that
\begin{align}
    \norm{(V_A\otimes V_B)\big(f_A(g)\otimes \Id_{\H_B}\ket{\psi}\big)-\big(\sigma(g)\otimes \Id_{\C^{2^n}}\ket{\eprn} \big)\otimes\ket{aux} }^2\leq O(n^{12})\sqrt{\epsilon}
\end{align}
for all $g\in H(n)$.
\end{thm}
\begin{proof}
By \Cref{polyn}, we see that
\begin{align*}
    &\norm{f_A(X_i)\otimes \Id-\Id\otimes f_B(X_i)}^2_\rho=2-2\bra{\psi}f_A(X_i)\otimes f_B(X_i)\ket{\psi}\leq O(n^6)\epsilon,\\
    &\norm{f_A(Z_i)\otimes \Id-\Id\otimes f_B(Z_i)}^2_\rho=2-2\bra{\psi}f_A(Z_i)\otimes f_B(Z_i)\ket{\psi}\leq O(n^6)\epsilon,\text{ and }\\
    &\norm{\big(f_A(X_i)f_A(Z_i)-f_A(Z_i)f_A(X_i)\big)\otimes \Id}^2_\rho=\norm{f_A(X_i)f_A(Z_i)-f_A(Z_i)f_A(X_i)}^2_{\rho_A}\leq O(n)\sqrt{\epsilon}
\end{align*}
for all $i\in[n]$. We obtain
\begin{align*}
    &\norm{f_A(X_i)f_A(Z_i)\otimes \Id-\Id\otimes f_B(X_i)f_B(Z_i)}_{\rho}\\
    =&\norm{\big(f_A(X_i)f_A(Z_i)-f_A(Z_i)f_A(X_i)\big)\otimes \Id +f_A(Z_i)f_A(X_i)\otimes \Id-\Id\otimes f_B(X_i)f_B(Z_i)}_{\rho}\\
    \leq&\norm{f_A(Z_i)f_A(X_i)\otimes \Id-\Id\otimes f_B(X_i)f_B(Z_i)}_\rho+O(n)\epsilon^{1/4}\\
    =&\norm{\big(f_A(Z_i)\otimes \Id\big)\big(f_A(X_i)\otimes \Id-\Id\otimes f_B(X_i)\big)-\big(\Id\otimes  f_B(X_i)\big)\big(\Id\otimes f_B(Z_i)-f_A(Z_i)\otimes \Id\big)}_\rho+O(\sqrt{n})\epsilon^{1/4}\\
    \leq& \norm{f_A(X_i)\otimes \Id-\Id\otimes f_B(X_i)}_\rho + \norm{\Id\otimes f_B(Z_i)-f_A(Z_i)\otimes \Id}_\rho+O(\sqrt{n})\epsilon^{1/4}\\
    \leq& O(n^3)\sqrt{\epsilon}+O(n^3)\sqrt{\epsilon}+O(\sqrt{n})\epsilon^{1/4}\leq O(n^3)\epsilon^{1/4}.
\end{align*}
This implies $\bra{\psi} f_A(X_i)f_A(Z_i)\otimes f_B(X_i)f_B(Z_i) \ket{\psi}=1-\frac{1}{2}\norm{f_A(X_i)f_A(Z_i)\otimes \Id-\Id\otimes f_B(X_i)f_B(Z_i)}^2_{\rho}\geq 1- O(n^6)\sqrt{\epsilon}$.
Hence we have
\begin{align}
    \bra{\psi}F_A\otimes F_B\ket{\psi}\geq 1-O(n^6)\sqrt{\epsilon}\label{FAFB}
\end{align}
for all $F_A\otimes F_B\in\{f_A(X_i)^af_B(Z_i)^b\otimes f_A(X_i)^af_B(Z_i)^b:a,b\in\{0,1\},i\in [n]  \}$.

By \Cref{thm:isometries}, there are Hilbert spaces $\mcH_A^{aux}$ and $\mcH_B^{aux}$, isometries $V_A:\mcH_A\arr\C^{2^n}\otimes \mcH_A^{aux}$ and $V_B:\mcH_B\arr\C^{2^n}\otimes \mcH_B^{aux}$, and representations $\phi_A(g):=\sigma(g)\otimes\Id_{\mcH_A^{aux}}$ and $\phi_B(g):=\sigma(g)\otimes\Id_{\mcH_B^{aux}}$ such that $\norm{V_Af_A(g)-\phi_A(g)V_A}_{\rho_A}^2\leq O(n^{12})\sqrt{\epsilon}$ and $ \norm{V_Bf_B(g)-\phi_B(g)V_B}_{\rho_B}^2\leq O(n^{12})\sqrt{\epsilon}$.

Let $\ket{\wtd{\psi}}:=V_A\otimes V_B\ket{\psi}\in(\C^{2^n}\otimes \mcH_A^{aux})\otimes(\C^{2^n}\otimes \mcH_B^{aux})$. Then for every  $(F_A\otimes F_B,\Phi_A\otimes \Phi_B)\in\{\big(f_A(X_i)^af_B(Z_i)^b\otimes f_A(X_i)^af_B(Z_i)^b,\phi_A(X_i)^a\phi_B(Z_i)^b\otimes\phi_A(X_i)^a\phi_B(Z_i)^b\big):a,b\in\{0,1\},i\in [n]  \}$, we see that
\begin{align*}
  &\abs{\bra{\psi}F_A\otimes F_B\ket{\psi}-\bra{\widetilde{\psi}}\Phi_A\otimes\Phi_B\ket{\widetilde{\psi}}}\\
  =&\frac{1}{2}\abs{\norm{V_AF_A\otimes V_B-V_A\otimes V_BF_B}^2_\rho-\norm{\Phi_AV_A\otimes V_B-V_A\otimes \Phi_B V_B }^2_\rho}\\
  =&\frac{1}{2}\big(\abs{\norm{V_AF_A\otimes V_B-V_A\otimes V_BF_B}_\rho+\norm{\Phi_AV_A\otimes V_B-V_A\otimes \Phi_B V_B }_\rho}\big)\abs{\norm{V_AF_A\otimes V_B\\
  &-V_A\otimes V_BF_B}_\rho-\norm{\Phi_AV_A\otimes V_B-V_A\otimes \Phi_B V_B }_\rho}\\
  \leq& O(n^3)\epsilon^{1/4}\norm{V_AF_A\otimes V_B
  -V_A\otimes V_BF_B-(\Phi_AV_A\otimes V_B-V_A\otimes \Phi_B V_B)  }_\rho\\
  \leq&O(n^3)\epsilon^{1/4}\big(\norm{V_AF_A\otimes V_B
  -\Phi_AV_A\otimes V_B}_\rho+ \norm{V_A\otimes V_BF_B-V_A\otimes \Phi_B V_B  }_\rho \big)\\
  =&O(n^3)\epsilon^{1/4}\big(\norm{V_AF_A
  -\Phi_AV_A}_{\rho_A}+ \norm{V_BF_B-\Phi_B V_B  }_{\rho_B} \big)
  \leq O(n^3)\epsilon^{1/4}O(n^6)\epsilon^{1/4}= O(n^9)\sqrt{\epsilon}.
\end{align*}
Then it follows from \Cref{FAFB} that
\begin{align}
    \bra{\widetilde{\psi}}\Phi_A\otimes\Phi_B\ket{\widetilde{\psi}}\geq 1-O(n^9)\sqrt{\epsilon}\label{Phi}
\end{align}
for all $\Phi_A\otimes \Phi_B\in\{\varphi_A(X_i)^a\varphi_B(Z_i)^b\otimes\varphi_A(X_i)^a\varphi_B(Z_i)^b:a,b\in\{0,1\},i\in [n]  \}$. Let $d:=\dim(\mcH_A^{aux}\otimes\mcH_B^{aux})$, and choose an orthonormal basis $\{\ket{\kappa_i}\}_{i=1}^d$ for $\mcH_A^{aux}\otimes \mcH_B^{aux}$. Then we can write $\ket{\widetilde{\psi}}=\sum_{j=1}^d\lambda_j\ket{\psi_j}\otimes\ket{\kappa_j}$ for some unit vectors $\ket{\psi_j}\in \C^{2^n}\otimes\C^{2^n}$ and $\lambda_j\in \C$ such that $\sum_j\abs{\lambda_j}^2=1$.
For every $j\in [d]$, let $\delta_{j}:=1-\frac{1}{4n}\sum_{a,b\in\{0,1\},i\in[n]}\bra{\psi_{j}}\big(\sigma_X(e_i)^a\sigma_Z(e_i)^b\otimes\sigma_X(e_i)^a\sigma_Z(e_i)^b \big)\ket{\psi_{j}}$.
 Then $\abs{\braket{\psi_{j}}{\eprn}}^2\geq 1-n\delta_{j}$ by \Cref{lemma:epr}. \Cref{Phi} implies $\sum_{j=1}^d\abs{\lambda_j}^2(1-\delta_j)\geq 1-O(n^9)\sqrt{\epsilon}$, or in other words, $\sum_{j}\abs{\lambda_{j}}^2\delta_{j}\leq O(n^9)\sqrt{\epsilon}$. Let $\ket{aux}:=\sum_{j=1}^d\lambda_j\ket{\kappa_j}\in \mcH_A^{aux}\otimes\mcH_B^{aux}$. Then
\begin{align*}
   \abs{\braket{\wtd{\psi}}{\eprn,aux}}=&\sum_{j=1}^d\abs{\lambda_{j}}^2\abs{\braket{\psi_{j}}{\eprn}}\geq \sum_{j=1}^d\abs{\lambda_{j}}^2\abs{\braket{\psi_{j}}{\phi_{2^n}}}^2\\ \geq& \sum_{j=1}^d\abs{\lambda_{j}}^2(1-n\delta_{j}) = 1-n\sum_{j=1}^d\abs{\lambda_{j}}^2\delta_{j}\geq 1-O(n^{10})\sqrt{\epsilon}.
\end{align*}
Since any complex phase can be absorbed in $\ket{aux}$, we may assume without loss of generality that $\braket{\wtd{\psi}}{\eprn,aux}\geq 0$. It follows that
\begin{align*}
   \norm{\ket{\wtd{\psi}}-\ket{\eprn,aux}}^2=2-2\braket{\wtd{\psi}}{\eprn,aux}\leq O(n^{10})\sqrt{\epsilon}.
\end{align*}
Then for any $g\in H(n)$, we have
\begin{align*}
    &\norm{(V_A\otimes V_B)\big(f_A(g)\otimes \Id_{\H_B}\ket{\psi}\big)-\big(\sigma(g)\otimes \Id_{\C^{2^n}}\ket{\eprn} \big)\otimes\ket{aux} }\\
    \leq& \norm{\big(V_Af_A(g)-\phi_A(g)V_A\big)\otimes V_B\ket{\psi}}+\norm{\phi_A(g)\otimes \Id_{\widetilde{\mcH}_B}(\ket{\widetilde{\psi}}-\ket{\eprn,aux})}\\
    = & \norm{ V_Af_A(g)-\phi_A(g)V_A }_{\rho_A}+\norm{\ket{\widetilde{\psi}}-\ket{\eprn,aux}}
    \leq  O(n^6)\epsilon^{1/4} + O(n^5)\epsilon^{1/4}\leq O(n^6)\epsilon^{1/4}.
\end{align*}
We conclude that $\norm{(V_A\otimes V_B)\big(f_A(g)\otimes \Id_{\H_B}\ket{\psi}\big)-\big(\sigma(g)\otimes \Id_{\C^{2^n}}\ket{\eprn} \big)\otimes\ket{aux} }^2\leq O(n^{12})\sqrt{\epsilon}$ for all $g\in H(n)$.
\end{proof}

Recall that for any $W\in\{X,Z\}^n$ and $a\in\{0,1\}^n$, we use $\sigma_W(a)$ to denote tensor product of Pauli operators $\otimes_{i=1}^n\sigma_{W_i}^{a_i}$. In the canonical perfect strategy for LWPBT, Alice performs $\sigma_W(a)$ on question $W(a)\in\mcI_A$. Since $\abs{a}\leq 6$ for all $W(a)\in \mcI_A$, as an immediate consequence of  \Cref{rigidity0} and \Cref{polyn} we have:

\begin{cor}\label{rigidity}
There exists a constant $C_{\mathrm{lw}}>0$ such that the following holds. For any $\epsilon>0$, $n\in\N$, and strategy $\mcS=(\tau^A,\tau^B,\ket{\psi}\in\mcH_A\otimes\mcH_B)$ for the $n$-qubit LWPBT with winning probability $1-\epsilon$, there are isometries $V_A:\H_A\rightarrow (\C^2)^{\otimes n}\otimes\H_A^{aux},V_B:\H_B\rightarrow(\C^2)^{\otimes n}\otimes\H_B^{aux}$ and a unit vector $\ket{aux}\in\H_A^{aux}\otimes\H_B^{aux}$ such that
\begin{align*}
   \norm{ (V_A\otimes V_B)\big(\tau^A(W(a))\otimes Id_{\H_B}\ket{\psi}\big)-\big(\sigma_W(a)\otimes Id_{\C^{2^n}}\ket{\eprn}\big)\otimes\ket{aux} }\leq C_{\mathrm{lw}}n^{6}\epsilon^{1/4}
\end{align*}
for all $W(a)\in\I_A$.
\end{cor}

\subsection{Related self-tests}\label{subsec:OtherTests}

We are aware of two other self-tests, 2-out-of-$n$ CHSH game \cite{CRSV18}, and the 2-out-of-$n$ magic square (MS) game~\cite{MNY22}, which can test for $n$ EPR pairs but which only require honest players to implement constant-sized Pauli measurements. To be more precise, the 2-out-of-$n$ CHSH game is a self-test for $n$ EPR pairs while the game given in \cite{MNY22} is studied in the commuting operator model and tests for a Hilbert space of dimension $2^n$, although the analysis of the 2-out-of-$n$ of MS game can be extended to the tensor product model as a test for EPR pairs. The 2-out-of-$n$ MS enjoys roughly $n^4 \sqrt{\epsilon}$ robustness and the 2-out-of-$n$ CHSH game enjoys $n \sqrt{\epsilon}$ robustness for testing the state.

In our setting, we require a self-test that can be won perfectly and further can also test for arbitrary braids of Pauli measurements up to weight~6. We suspect that  similar such self-testing properties can be shown for a modified version of the 2-out-of-$n$ MS, which includes additional consistency checks. We believe a group theoretic approach using  \Cref{thm:GH} could be used to confirm these additional self-testing properties but this would likely entail a similar technical treatment as that given in~\Cref{subsec:ProofOfRigidity}. Both LWPBT and a modified 2-out-of-$n$ MS have $poly(n)$ many questions and a constant number of answers and have $poly(n) \sqrt{\epsilon}$ robustness.

%% file: HamiltonianGame.tex
In this section, we show that for some local Hamiltonian $H$, one can construct a nonlocal game $\mcG(H)$ whose winning probability is closely related to the ground state energy $\lambda_0(H)$ of $H$. Our game is based on the Hamiltonian game introduced by Grilo \cite{Gri19}, but some subtle reconfigurations are made for the purpose of the zero-knowledge property. In particular, instead of using Pauli braiding test as in \cite{Gri19}, we employ the low-weight Pauli braiding test against dishonest quantum provers and then perform parallel repetition to achieve a constant completeness-soundness gap. To incorporate LWPBT in our modified Hamiltonian test, we consider Hamiltonians with specific structures:

\begin{defn}

We say a Hamiltonian $H$ is of $XZ$-type if it can be decomposed as $H= \frac{1}{m}\sum_{\ell=1}^m \gamma_\ell H_\ell$ where each $\gamma_\ell \in [-1,1]$ and each term $H_\ell$ is a tensor product of operators $\sigma_X, \sigma_Z$ or $\sigma_I$.
\end{defn}

Next, we define the relevant energy test which is analogous to the energy test used in \cite{Gri19}.

\begin{defn}[Energy test]\label{et}
Given an $n$-qubit $6$-local Hamiltonian $H= \frac{1}{m}\sum^m_{\ell=1} \gamma_\ell H_\ell$ of $XZ$-type we define the following energy test:
\begin{enumerate}
    \item The verifier picks a term $H_\ell$ for $\ell \in [m]$ taken uniformly at random, and selects uniformly at random from the pairs
    $\{(W,r)\in \{X,Z\}^n\times \{0,1\}^n: \sigma_{W}(r)=H_\ell \}$.
    \item The verifier sends $W(r)$ to Alice, and tells Bob that the players are playing the energy test.
    \item Alice  responds with a single value $c \in \{ -1 , 1 \}$ and Bob responds with $2n$ bits $a_1, \dots a_n ,b_1, \dots b_n$.
    \item The verifier next computes bit string $d$ as follows. Take $d_i = (-1)^{a_{i}}$ if $r_i=1$ and $W_i=X$,  take $d_i=(-1)^{b_i}$ if $r_i=1$ and $W_i=Z$, and  take $d_i=0$ in all other cases.
    \item The verifier accepts if $c \cdot \prod_i d_i \neq sign(\gamma_l)$, and rejects with probability $\abs{\gamma_l}$ otherwise.
\end{enumerate}
\end{defn}

Combining the LWPBT and the above energy test we define our modified Hamiltonian test:

\begin{defn}[Hamiltonian test]\label{Hamiltonian Test}
Let $H= \frac{1}{m}\sum^m_{\ell=1} \gamma_\ell H_\ell$ be a $k$-local Hamiltonian of $XZ$-type and let $p\in(0,1)$. We define the following game $\mcG(H,p)$: with probability $(1-p)$ the players play LWPBT introduced in Section \Cref{Sec:LWPBT}, and with probability $p$ the players play energy test described in \Cref{et}.
\end{defn}

We refer to a strategy $\mcS$ for $\mcG(H,p)$ as a \emph{semi-honest strategy} if the players employ the canonical perfect strategy when playing LWPBT. Hence in a semi-honest strategy Alice and Bob hold $n$ EPR pairs and Alice must perform $\sigma_W(r)$ on question $W(r)$ since she cannot distinguish questions from LWPBT or energy test. We also define the \emph{honest strategy} $\mcS_h$ for $\mcG(H,p)$ in which the players employ the canonical perfect strategy when playing LWPBT, and in the energy test, Bob honestly teleports the ground state of $H$ to Alice and provides the verifier with the teleportation keys.

Below we analyze the players' ability to win the overall game $\mcG(H,p)$ assuming the players are using a semi-honest strategy.

\begin{lemma}[Lower bound on semi-honest strategies]\label{lower-boung lemma}
Suppose $H= \frac{1}{m}\sum^m_{l=1} \gamma_l H_l$ is an $n$-qubit $6$-local XZ Hamiltonian, and Alice and Bob are performing a semi-honest strategy $\mcS$ for $\mcG(H,p)$. Then
\begin{align*}
   \omega(\mcS)\leq \omega(\mcS_h) =1-p(\frac{1}{2m}\sum_l |\gamma_l| +\frac{1}{2}\lambda_0(H)).
\end{align*}
\end{lemma}
\begin{proof}
Suppose the players are employing a semi-honest strategy $\mcS=(\tau^A,\tau^B,\ket{\psi}\in \mcH_A\otimes\mcH_B)$ for $\mcG(H,p)$. Let $\tau:=\Tr_{\mcH_B}(\ket{\psi}\bra{\psi})$. Since the players win the LWPBT perfectly, they can only lose the overall game if they are playing an instance of the energy test. Let $a,b\in\{0,1\}^n$ be the answers Bob provides in the energy test, and let $\rho:=\sigma_X^b  \sigma_Z^a\tau  \sigma_Z^a \sigma_X^b$.

Suppose in a round of the energy test, the verifier picks an $\ell\in[m]$ and selects a $W(r)$ for Alice. As discussed above, since Alice cannot distinguish questions from LWPBT and the energy test, she must perform $\sigma_{W}(r)=H_\ell$ on her registers. Hence $
    \mathbb{E}(c \cdot\prod_i d_i)
    =\mathrm{Tr}(H_\ell \sigma_X^b  \sigma_Z^a\tau  \sigma_Z^a \sigma_X^b)
    =\mathrm{Tr}(H_\ell \rho)$. Let $p_\ell$ be the probability of $c\prod_id_i=sign(\gamma_\ell)$. Then $\mathbb{E}(c\prod_id_i)=p_\ell sign(\gamma_\ell)-(1-p_\ell)sign(\gamma_\ell)$, or in other words, $\gamma_\ell\mathbb{E}(c\prod_id_i)=(2p_\ell-1)\abs{\gamma_\ell}$. This implies the verifier rejects with probability

\begin{align*}
    p_\ell\abs{\gamma_\ell}=\frac{\abs{\gamma_\ell}+\gamma_\ell\mathbb{E}(c\prod_id_i)}{2}=\frac{\abs{\gamma_\ell}+\gamma_\ell\mathrm{Tr}(H_\ell \rho)}{2}.
\end{align*}

Thus by averaging over $\ell \in [m]$ we see that the players lose the energy test with probability
\begin{align*}
    \frac{1}{m}\sum_{\ell\in[m]} \frac{\abs{\gamma_\ell} + \gamma_\ell \mathrm{Tr}(H_\ell \rho)}{2}=\frac{1}{2m}\sum_{\ell\in[m]}\abs{\gamma_\ell}+\frac{1}{2}\mathrm{Tr}(H\rho).
\end{align*}
This probability is minimized if and only if $\rho$ is indeed the density matrix of the ground state of $H$ and in such case, the probability of winning the overall game is at most
\begin{align*}
    1-p\big(\frac{1}{2m} \sum_{\ell\in[m]} \abs{\gamma_\ell} +\frac{1}{2} \lambda_0(H)\big).
\end{align*}
This probability can be achieved if Bob teleports over the ground state and supplies the verifier with the verification keys in the energy test.
\end{proof}

\begin{lemma}(Upper bound on dishonest strategies)\label{Upper-Bound Lemma}
Let $H=\frac{1}{m}\sum^m_{\ell=1} \gamma_\ell H_\ell$ be a $6$-local, $n$-qubit Hamiltonian of $XZ$-type. For any $\eta\in(0,1)$, let $p=\frac{4n^{-6}\eta^{3/4}}{(C_{\mathrm{lw}}+1)3^{3/4}}$ where $C_{\mathrm{lw}}$ is the constant given in \Cref{rigidity}. Then
\begin{align*}
    \omega^*\big(\mcG(H,p) \big) \leq \omega(\mcS_h) + \eta,
\end{align*}
where $\omega(\mcS_h)=1-p\big(\frac{1}{2m} \sum_{\ell\in[m]} \abs{\gamma_\ell} +\frac{1}{2} \lambda_0(H)\big)$ as in \Cref{Upper-Bound Lemma}.
\end{lemma}

\begin{proof}
Suppose the provers are employing a strategy $\mcS=(\tau^A,\tau^B,\ket{\psi})$ for $\mcG(H,p)$ that wins LWPBT with probability $1-\epsilon$ and wins the energy test with probability $\delta +1-\frac{\sum_\ell \abs{\gamma_\ell}}{2m} - \frac{\lambda_0(H)}{2}$. \Cref{rigidity} implies $\delta \leq C_{\mathrm{lw}}n^6\epsilon^{1/4}$. Then for $p:=\frac{4n^{-6}\eta^{3/4}}{(C_{\mathrm{lw}}+1)3^{3/4}}$ with $\eta\in(0,1)$, we have $\eta+\epsilon=\eta/3+\eta/3+\eta/3+\epsilon \geq 4(\frac{\eta^3\epsilon}{3^3})^{1/4}=p(C_{\mathrm{lw}}+1)n^6\epsilon^{1/4}$. It follows that
\begin{align*}
    p\delta-(1-p)\epsilon\leq pCn^6\epsilon^{1/4}+p\epsilon-\epsilon\leq pC_{\mathrm{lw}}n^6\epsilon^{1/4}+pn^6\epsilon^{1/4}-\epsilon
    =p(C_{\mathrm{lw}}+1)n^6\epsilon^{1/4}-\epsilon\leq \eta.
\end{align*}
Hence the overall winning probability is given by
\begin{align*}
    \omega(\mcS) =(1-p)(1-\epsilon) + p(\delta +1-\frac{\sum_\ell \abs{\gamma_\ell}}{2m} - \frac{\lambda_0(H)}{2})
    = \omega(\mcS_h) +p \delta -(1-p)\epsilon\leq \omega(\mcS_h)+\eta.
\end{align*}
\end{proof}

In the rest of this paper, given an $n$-qubit, $6$-local Hamiltonian $H$ of $XZ$-type and parameters $\alpha$ and $ \beta$ with $\beta-\alpha\geq 1/poly(n)$, we use $\mcG(H)$ to denote the game $\mcG(H,p)$ with $p=\frac{32n^{-6}(\beta-\alpha)^{24}}{27(C_{\mathrm{lw}}+1)^4}$.

\begin{thm}\label{thm:Probabilitiy choice}
Given an $n$-qubit, 6-local Hamiltonian $H= \frac{1}{m}\sum^m_{\ell=1} \gamma_\ell H_\ell$ and parameters $\alpha, \beta$ with $\beta-\alpha\geq 1/poly(n)$, let $\omega_{\alpha}$ (resp. $\omega_{\beta}$) denote the maximum winning probability for $\mcG(H)$ when $\lambda_0(H) \leq \alpha$ (resp. $\lambda_0(H) \geq \beta$). Then  $\omega_{\alpha} -\omega_{\beta} \geq 1/poly(n)$.
\end{thm}

\begin{proof}
Let $\eta:=\frac{16(\beta-\alpha)^{32}}{27(C_{\mathrm{lw}}+1)^4}$, and let  $p:=\frac{4n^{-6}\eta^{3/4}}{(C_{\mathrm{lw}}+1)3^{3/4}}$. Then $p=\frac{32n^{-6}(\beta-\alpha)^{24}}{27(C_{\mathrm{lw}}+1)^4}$, and hence $\mcG(H)=\mcG(H,p)$.  By \Cref{lower-boung lemma} and \Cref{Upper-Bound Lemma} we have $\omega(\mcS_h)\leq\omega^*\big(\mcG(H,p) \big) \leq \omega(\mcS_h) + \eta$. This implies $\omega_\alpha\geq 1-p(\frac{1}{2m}\sum_\ell\abs{\gamma_\ell}+\frac{1}{2}\alpha)$ and $\omega_\beta\leq 1-p(\frac{1}{2m}\sum_\ell\abs{\gamma_\ell}+\frac{1}{2}\beta)+\eta$. Since $n^6(\beta-\alpha)^7\leq O(n^{-1})$, it follows that
\begin{align*}
    \omega_\alpha-\omega_\beta\geq \frac{1}{2}p(\beta-\alpha)-\eta=\frac{16n^{-6}(\beta-\alpha)^{25}}{27(C_{\mathrm{lw}}+1)^4}(1-n^6(\beta-\alpha)^7)\geq \frac{16n^{-6}(\beta-\alpha)^{25}}{27(C_{\mathrm{lw}}+1)^4}\geq \frac{1}{poly(n)}.
\end{align*}
\end{proof}

Our approach to gap amplification follows that of \cite{GSY19} which uses parallel repetition for the anchored version of games with perfect completeness. The key difference is that in our case the underlying games do not satisfy perfect completeness. Instead we apply a threshold parallel repetition theorem due to Yuen to the anchored version of our games.

For every $n\in\N$, let $\mcG(H),w_\alpha $ and $w_\beta$ be as in \Cref{thm:Probabilitiy choice}. By \cite[Theorem~41]{Yue16}, there exists a $poly(n)$-computable transformation, called \emph{anchoring}, that transforms $\mcG(H)$ to a two-player game $\mcG(H)_{\perp}$ with winning probability $1-\tfrac{1-w^*(\mcG(H))}{2}$. So $w^*(\mcG(H)_\perp)=\begin{cases}
   1-\epsilon_\alpha/2 & \text{ if } \lambda_0(H)\leq \alpha\\
   1-\epsilon_\beta/2 & \text{ if } \lambda_0(H)\geq \beta
\end{cases}$, where $\epsilon_\alpha:=1-w_\alpha$ and $\epsilon_\beta:=1-w_\beta$. Then by \cite[Theorem~42]{Yue16}, there is a universal constant $C>0$ such that for all integer $m\geq 1$, and $\gamma \geq 0$, the probability that in the game $\mcG(H)_{\perp}^m$ the players can win more than $\big(w^*(\mcG(H)_\perp)+\gamma\big)m$ games is at most $(1-\gamma^9/2)^{Cm}$. Take $\gamma:=\tfrac{\epsilon_\alpha-\epsilon_\beta}{4}$ and $m=\max\{4\gamma^{-2},2C\gamma^{-9}\}$. Let $\widehat{\mcG}(H):=\mcG(H)_{\perp}^m$ be the $m$ parallel repeated anchoring version of $\mcG(H)$. We show that this nonlocal game has a constant completeness-soundness gap.
\begin{thm}\label{gap}
    Let $\widehat{\omega}_\alpha$ (resp. $\widehat{\omega}_\beta$) be the maximum winning probability for $\widehat{\mcG}(H)$ when $\lambda_0(H)\leq \alpha$ (resp. $\lambda_0(H)\geq \beta$). Then $\widehat{\omega}_\alpha-\widehat{\omega}_\beta\geq 1/4$.
    
\end{thm}
\begin{proof}
    If $\lambda_0(H)\geq \beta$, then $\widehat{\omega}_\beta\leq (1-\tfrac{\gamma^9}{2})^{Cm}\leq (1-\tfrac{\gamma^9}{2})^{2/\gamma^{9}}< e^{-1}<1/2$. Now suppose $\lambda_0(H)\leq \alpha$. An optimal strategy $S$ for $\mcG(H)_\perp$ has winning probability $1-\tfrac{\epsilon_\alpha}{2}$. Let $X$ be the random variable for the number of games the strategy $S^m$ wins. Then $X\sim\text{Binomial}(m,1-\tfrac{\epsilon_\alpha}{2})$, so $\mathbb{E}X=m(1-\tfrac{\epsilon_\alpha}{2})$ and $\text{Var}X=m\tfrac{\epsilon_\alpha}{2}(1-\tfrac{\epsilon_\alpha}{2})$. Since $(1-\tfrac{\epsilon_\alpha}{2})-(1-\tfrac{\epsilon_\beta}{2}+\gamma)=\frac{\epsilon_\beta-\epsilon_\alpha}{2}-\gamma=\gamma$, we obtain that
\begin{align*}
    Pr(X\leq (1-\tfrac{\epsilon_\beta}{2}+\gamma)m)&=Pr((1-\tfrac{\epsilon_\alpha}{2})m-X\geq \gamma m)\\
    &\leq Pr(\abs{X-\mathbb{E}X}\geq \gamma m)\\
    &\leq \frac{m(1-\tfrac{\epsilon_\alpha}{2})\tfrac{\epsilon_\alpha}{2}}{(\gamma m)^2}=\frac{1}{m\gamma^2}\leq 1/4.
\end{align*}
This implies $\widehat{\omega}_\alpha\geq w(S^m)=1-Pr(X\leq (1-\tfrac{\epsilon_\beta}{2}+\gamma)m)\geq 3/4$, so the theorem follows.
\end{proof}

%% file: ZK.tex
In this section, we show that the family of games described in \Cref{Hamiltonian Test} provides a statistical zero-knowledge $\mathrm{MIP}^*[2,1]$ protocol for $\mathrm{QMA}$ with inverse polynomial completeness/soundness gap. First we introduce our protocol.

\subsection{Simulation of history states for $XZ$-Hamiltonians}
Before we introduce our $\MIP^*$ protocol and proceed to our result on zero-knowledge, we reformulate a result, originally introduced by Broadbent and Grilo~\cite{BG22} (Lemma 3.5), so that it is more amenable to device-independent techniques.

\begin{thm}[Simulation of history states]\label{thm:simulation of history states}
For any language $L=(L_{yes},L_{no})$ in $\mathrm{QMA}$ and $s \in \mathbb{N}$, there is a family of verification circuits $V^{(s)}_{x}=U_T \dots U_1$ for $L$ that acts on a witness of size $p(|x|)$ and on $q(|x|)$ ancillary qubits such that there exists a polynomial-time deterministic algorithm $Sim_{V^{(s)}}$ that takes as input an instance $x \in L$ and a subset $S \subseteq [T+p+q]$ with $|S| \leq 3s+2$, then outputs a classical description of an $|S|$-qubit density matrix $\rho(x,S)$ with the following properties:
\begin{enumerate}
    \item If $x\in L_{yes}$, then there exists a $p(\abs{x})$-qubit witness $\psi^{s}$ such that $V^{(s)}_{x}$ accept with probability at least $1-negl(n)$ on $\psi^{s}$ and $\| \rho(x,S) - \mathrm{Tr}_{\overline{S}}(\rho) \|_{tr} \leq negl(|x|)$, where \[ \rho = \frac{1}{T+1} \sum_{t,t' \in [T+1]} \ket{unary(t)}\bra{unary(t)} \otimes U_{t}\dots U_1(\psi^{s} \otimes \ket{0}\bra{0}^{\otimes q})U^{*}_1\dots U^{*}_{t'} \] is the history state of $V^{(s)}_{x}$ on witness $\psi^s$.

\item Let $H_i$ be one of the terms from the circuit-to-local Hamiltonian construction from $V^{(s)}_{x}$, and let $S_i$ be the set of qubits on which $H_i$ acts non-trivially. Then $\mathrm{Tr}(H_i\rho(x,S_i))=0$ for all $x \in L$.
\item The Hamiltonian $H$ from the circuit-to-local Hamiltonian construction is a $6$-local Hamiltonian of $XZ$ type.
\end{enumerate}
\end{thm}

\begin{proof}
The first two points were proven by Broadbent and Grilo using simulatable codes constructed from a different set of physical gates \cite{BG22}. In brief, Broadbent and Grilo provide a constructive approach which for each $x \in L$ provides a transformation from the original verification circuit $V$ to a simulatable verification circuit $V^{(s)}_{x}$ that has $4$ main stages:
\begin{enumerate}[label=(\roman*)]
    \item The verifier receives an encoded witness and checks if the provided witness is properly encoded under the relevant code $\mathcal{C}$,
    \item creates  auxiliary resource states $\ket{0}, \ket{T}$ and encodes them under $\mathcal{C}$,
    \item performs an encoded version of the original verification circuit $V$ on the encoded witness using either transversal Clifford gates or suitable encodings of $T$-gates, which use $\ket{T}$, and
    \item  decodes the output of the computation.
\end{enumerate}
For each constant $s$ the CECC $\mathcal{C}$ is taken to be $\log(3s+5)$-fold concatenated Steane code.

\Cref{thm:sim-codes} is obtained by using the same code $\mathcal{C}$ and applying the $3$-qubit Toffoli gate using resource state $\ket{\text{Toffoli}}$, which can be constructed from $\ket{000}$ using $H$ and $\Lambda^2(X)$. The proof of our result then follows by appropriately changing the resource generation stage $(ii)$ to create encodings of auxiliary states $\ket{0}$ and $\ket{\text{Toffoli}}$. Then stage $(iii)$ is altered to instead apply an encoding of the original circuit using encodings of the gate set $\lbrace H, \Lambda(X), \Lambda^2(X) \rbrace $.

To see how this change of gate set gives us the third result, note that for any gate $G$ we have $\Lambda(G)=\frac{1}{2}(I+Z)\otimes I + \frac{1}{2}(I-Z)\otimes G$. Using this decomposition we can readily check that gates $\Lambda(X)$ and $\Lambda^2(X)$ can both be expressed as a real linear combination of tensor products of $X,Z$, and $I$. The same holds for $H= \frac{1}{\sqrt{2}} (X+Z)$. Consequently, all physical gates used in stages $(ii)$ and $(iii)$ will consist of gates that can be expressed as real linear combinations of tensors of $X$ and~$Z$. The same is true for physical gates used in stages $i)$ and $iv)$. We will also have that all of these gates are of order~$2$. Finally, following the same approach due to Ji (Lemma 22 ~\cite{Ji16}),  we get $H$ is $6$-local and of $XZ$ type.
\end{proof}

Below we only need to invoke \Cref{thm:simulation of history states} for the case of $s=2$ in order to our zero-knowledge protocol. We use $V_x$ to denote $V_x^{(2)}$ throughout the rest of this section.

\subsection{A two prover zero-knowledge proof system for ${\QMA}$}

Let $L=(L_{yes}, L_{no})$ be a language in $\QMA$. \Cref{xtoGhat} describes a two-prover one-round interactive proof system for $L$ with a constant polynomial completeness-soundness
gap.

\begin{figure}[h]
    \begin{mdframed}
    \centering
    $x\xrightarrow{\text{\Cref{thm:simulation of history states}}} V_x\xrightarrow{\text{circuit-to-Hamiltonian}} H_x\xrightarrow{\text{\Cref{Hamiltonian Test}}}\mcG_x:=\mcG(H_x) \xrightarrow{\text{\Cref{gap}}} \widehat{\mcG}_x:=\widehat{\mathcal{G}}(H_x)$
    \end{mdframed}
    \caption{$x$ is an instance in $L \in \QMA$. $V_x$ is a $poly(\abs{x})$-size quantum circuit. $H_x$ is a $poly(\abs{x})$-qubit 6-local Hamiltonian of $XZ$-type. $\widehat{\mathcal{G}}_x$ is a nonlocal game with $poly(\abs{x})$-bit questions and $poly(\abs{x})$-bit answers.}
    \label{xtoGhat}
\end{figure}
To put it more concretely, for any $x\in L$, applying \Cref{thm:simulation of history states}, with $s=2$, provides a $poly(\abs{x})$-size verification circuit $V_x$, and the circuit-to-Hamiltonian construction creates a $poly(\abs{x})$-qubit 6-local Hamiltonian $H_x$ of $XZ$-type. We can then construct a non-local game $\widehat{\mcG}_x$ as in \Cref{sec:Game}. The constant completeness-soundness gap follows from \Cref{gap}. The map from instances $x\in L$ to nonlocal games $\widehat{\mcG}_x$ is efficient since each step described in \Cref{xtoGhat} is efficient.

To prove the above interactive proof system for $L$ has the statistical zero-knowledge property, we first establish that any malicious verifier $\widehat{V}$ and  $x \in L_{yes}$, there exists a PPT simulator that can sample from $View(\widehat{V}(x), \mathcal{S}_h)$, where $\mcS_h$ is the honest strategy for $\mcG_x$ defined in \Cref{sec:Game}.

\begin{lemma}
Suppose $x \in L_{yes}$ for some language $L=(L_{yes},L_{no})$ in $\QMA$. Let $\mcG_x$ be the corresponding nonlocal game described in \Cref{xtoGhat}, and let $\mcS_h$ be the honest strategy for $\mcG_x$ defined in \Cref{sec:Game}. For any malicious verifier $\widehat{V}$ there exists a PPT algorithm $Sim_{\widehat{V}}$ with output distribution negligibly close to $View(\widehat{V}(x), \mathcal{S}_h)$, 
\end{lemma}

\begin{proof}
Let $H_x$ be the corresponding Hamiltonian, and suppose $H_x$ acts on $n$ qubits ($n=poly(\abs{x})$). Recall that, in $\mcS_h$, Alice and Bob measure on a shared $n$-EPR pair $\ket{\eprn} \in \mathcal{H}_A \otimes \mathcal{H}_B$. Bob has an addition $n$-qubit register $\mcH_{B'}$ holding a copy of the ground state $\rho$ of $H_x$. Let $\mathcal{H}_{A,i}, \mathcal{H}_{B,i}$ and $\mathcal{H}_{B',i}$ be the $i$-th qubit of Alice or Bob on their respective registers.

As discussed in \Cref{Sec:PrelimsZk} a malicious verifier $\widehat{V}$ samples its first question $q_1$ using some distribution on the set of all possible questions for Alice or Bob. Given a response $r_1$ the second question $q_2$ is conditionally sampled using distribution. Since $\widehat{V}$ is $\mathrm{PPT}$ the simulator $Sim_{\widehat{V}}$ can also sample from these distributions given the same randomness. We consider the cases below for conditionally generating the responses.

If neither question $q_1$ nor $q_2$ are requests for energy test then in the honest strategy both players will make one or more Pauli measurements on their shared $n$-EPR pairs, and hence the simulator can efficiently produce responses. Here the order of the players will not matter as the simulator can first randomly generate the response~$r_1$ of the first player and compute the appropriate response~$r_2$ of the second conditional on the previous entries in the transcript.

Suppose $q_1$ is a request for an energy test sent to Bob. In the honest strategy, Bob performs $n$ Bell measurements on $\mathcal{H}_B \otimes \mathcal{H}_{B'}$. The simulator can produce a response $r_1$ conditional on Bob receiving a request for energy test by sampling $2n$ uniformly random bits $(\alpha ,\beta) =(\alpha_1,\cdots,\alpha_n, \beta_1,\cdots,\beta_n)$. In the honest strategy Alice should now have $\rho_{\alpha, \beta} :=\sigma_Z^{\beta}\sigma_X^{\alpha}\rho \sigma_X^{\alpha} \sigma_Z^{\beta}$ on her register $\mathcal{H}_A$, where $\sigma_W^r$ denotes $\otimes_{i\in [n]}\sigma_W^{r_i}$ for $W\in\{X,Z\}$ and $r\in\{0,1\}^n$.

It follows that $q_2$ must be a measurement request for Alice. This can be either of the form  $\sigma_{W}(a)$ for some $W \in \lbrace X,Z,I \rbrace^n$ and $a \in \lbrace 0,1 \rbrace^n$ with $|a| \leq 6$, or $\sigma_{XZ} \otimes \sigma_{ZX}(a)$ with $|a|=2$.
If $S$ is the set of qubits on which Alice's observable acts non-trivially then by \Cref{thm:simulation of history states} there exists a polynomial time algorithm which can output a description of a density matrix $\rho(x,S)$ satisfying
\begin{align}
    \norm{\mathrm{Tr}_{\overline{S}}(\rho) - \rho(x,S) }_{tr} \leq negl(n).\label{eq:negl}
\end{align}
$Sim_{\widehat{V}}$ can then append $r_2$ to the transcript using an approximation for Alice's honest measurement result on the qubits in $S$ of $\rho_{\alpha, \beta}$. This approximation can then be calculated using $\rho(x,S)$ since
\begin{equation*}
   \mathrm{Tr}_{\overline{S}}(\sigma_Z^{\beta}\sigma_X^{\alpha}\rho \sigma_X^{\alpha} \sigma_Z^{\beta}) = \big(\bigotimes_{i \in S}\sigma_{Z}^{b_{i}}\sigma_{X}^{a_{i}}\big) \mathrm{Tr}_{\overline{S}}(\rho) \big(\bigotimes_{i \in S}\sigma_{Z}^{b_{i}}\sigma_{X}^{a{i}}\big).
\end{equation*}
Lastly, we consider the case that Alice is first given a measurement request $q_1$ and after receiving response $r_1$ the verifier sends an energy test request to Bob for $q_2$. In order to generate a response for Alice the simulator will first generate a uniformly random response $r_1$. We take $\Phi_{q_{1},r_{1}}$ denote the post measurement state on $\mathcal{H}_A \otimes \mathcal{H}_B$ obtained by measuring the shared EPR pairs according to $q_1$ and obtaining outcome $r_1$. In the honest strategy, the overall shared state will be
\begin{align*}
    \tau_{q_{1},r_{1}} := \Phi_{q_{1},r_{1}} \otimes \rho \in \mathcal{H}_A \otimes \mathcal{H}_B \otimes \mathcal{H}_{B'},
\end{align*}
and Bob will perform independent Bell measurements on the registers $\mathcal{H}_B \otimes \mathcal{H}_{B'}$. Again we take~$S$ to denote the set of qubits on which Alice measured given question $q_1$. Then we have
\begin{align*}
   \mathrm{Tr}_{\overline{S}}(\tau_{q_{1}, r_{1}}) = \mathrm{Tr}_{\overline{S}}(\Phi_{q_{1}, r_{1}}) \otimes \mathrm{Tr}_{\overline{S}}(\rho).
\end{align*}
It follows from \Cref{eq:negl} that
\begin{align*}
    \norm{\mathrm{Tr}_{\overline{S}}(\tau_{q_{1}, r_{1}})-\mathrm{Tr}_{\overline{S}}(\Phi_{q_{1}, r_{1}})\otimes \rho(x,S)}_{tr}\leq negl(n).
\end{align*}
Thus the simulator can generate the conditional response $(\alpha, \beta)=(\alpha_i,\beta_i)_i$ as follows. First compute reduced density matrix $\rho(x,S)$ from \Cref{thm:simulation of history states}. For $i \in S$ sample $(\alpha_i, \beta_i)$ according to the probability of each possible Bell measurement result on $\mathrm{Tr}_{\overline{S}}(\Phi_{q_{1}, r_{1}}) \otimes \rho(x,S)$.
For the remaining $i \notin S$ the simulator can sample uniformly random pairs $(\alpha_i, \beta_i)$.

In all of the above cases, the output of the simulator will be guaranteed negligibly close and thus the output of $Sim_{\widehat{V}}$ is negligibly close to $View(x,\mathcal{S})$.
\end{proof}

All that remains is to argue that the interactive protocol described in \Cref{xtoGhat} based on the scaled-up game $\widehat{\mcG}_x$ is statistically zero-knowledge.

\begin{thm}\label{thm:main}
    The protocol described in \Cref{xtoGhat} is statistical zero-knowledge and has a constant completeness-soundness gap.
\end{thm}
\begin{proof}
    The constant completeness-soundness gap follows directly from \Cref{gap}. To show the statistical zero-knowledge, we first consider the anchoring procedure for the game $\mcG_x$. We can specify an honest strategy $\mcS_{h,\perp}$ for the anchored version of $\mcG_x$ by fixing a choice of output for either player who receives question $\perp$ in the honest strategy. Then, given any malicious verifier $\widehat{V}(x)$, the simulator given in \Cref{thm:main} can be trivially modified to sample from a distribution which is negligibly close to $View(\widehat{V}(x), \mathcal{S}_{h, \perp})$.

In the case of the threshold parallel repeated game $\widehat{\mcG}_x$, the honest strategy $\mcS^{m}_{h, \perp}$ is taken to be the $m$-fold product of the honest strategy $\mcS_{h, \perp}$. Then, as commented in \cite{GSY19}, since the protocol only queries each player once, a new simulator can be obtained by sampling according to the $m$-fold product of the simulator used in the above lemma.
\end{proof}

%% file: SkepticalVerifierModel.tex
In this section, we formally define the OTS model (\Cref{subsec:FormalModel}). In \Cref{subsec:Applications}, we state our main conceptual theorems, which make use of all the technical contributions in the main part of this work.

\subsection{Formal description of the model}\label{subsec:FormalModel}

Here we provide a formal description of the OTS model. This model is defined as a refinement of ${\MIP}^*$, where the completeness condition is weakened, allowing  only one of the provers to be ``all-powerful", while the other has limited functionality determined independently of the problem instance.

\paragraph{Off-the-shelf device.} We first formalize the definition of a family of off-the-shelf devices. A \emph{verification device} $D=( \ket{\psi}, \lbrace P^{1}_{a}\rbrace_a, \dots, \lbrace P^{q}_{a} \rbrace_a )$ consists of a state $\ket{\psi}$ on Hilbert spaces $\mathcal{K}_A \otimes \mathcal{K}_B$ and a collection of  POVMs $ \lbrace P^{1}_{a}\rbrace_a, \dots, \lbrace P^{q}_{a} \rbrace_{a}$ on $\mathcal{K}_A$. Recall that a quantum strategy for a non-local game is determined by a tuple of the form $\mcS=(\{E^x_a\}_{x,a},\{F^y_b \}_{y,b},\ket{\phi} \in \mathcal{H}_A \otimes \mathcal{H}_B)$ as outlined in \Cref{sec:prelims}. We say that such a quantum strategy $\mathcal{S}$ can be \emph{implemented} using $D$, if three conditions hold: (i) $\mathcal{H}_A = \mathcal{K}_A$ and $\mathcal{H}_B = \mathcal{K}_B \otimes \mathcal{K}_{B'}$ for some Hilbert space $\mathcal{K}_{B'}$. (ii) The set of measurements in $\mathcal{S}$ which act on $\mathcal{H}_A$ are a contained in $D$. (iii) The shared state $\ket{\phi}$ in $\mathcal{S}$ can be decomposed as $\ket{\phi} = \ket{\psi} \otimes \ket{\phi'}$ where $\ket{\psi} $ is the state in $D$ and $\ket{\phi'}$  is some auxiliary state held on a Hilbert space $\mathcal{K}_{B'}$.

Note that if a quantum strategy $\mathcal{S}$ can be implemented using a verification device $D$, then the overall system is decomposed as $\mathcal{H}_A \otimes \mathcal{H}_B = \mathcal{K}_A \otimes \mathcal{K}_B \otimes \mathcal{K}_{B'}$. Informally, we view this decomposition as consisting of an entangled state $\ket{\psi}$ shared between the device and a single prover, together with an auxiliary state $\ket{\phi'}$ accessible only to the prover.

Given a collection of verification devices $\lbrace D_n \rbrace_{n \in \mathbb{N}}$, where each $D_n$ consists of a state $\ket{\psi_n}$ and a sequence of POVMs, we say $\{D_n\}_{n\in\N}$ is an \emph{efficient family of off-the-shelf devices} if there exists a polynomial-time uniform family of quantum circuits $\lbrace Q_n \rbrace_{n \in \mathbb{N}} $ satisfying the following: $Q_n$ generates the state $\ket{\psi_n}$ from an all $0$ state, and on input $i$ measures $\ket{\psi_n}$ using the $i$-th POVM from $D_n$.

\begin{defn}\label{defn:SQVM-(Formal)}
A promise language $L=( L_{yes},L_{no})$ has an off-the-shelf (OTS) proof system if there exists an efficient family of off-the-shelf devices $\lbrace D_n \rbrace_{n\in\N}$, and a  polynomial-time computable function that takes an instance $x$ to the description of a non-local game $\mcG_x$ satisfying the following:
    \begin{enumerate}
        \item \textbf{Completeness using off-the-shelf devices.} For any $x \in L_{yes}$ with $|x| \leq n$, there exists a quantum strategy $\mcS_x$, which can be implemented using $D_n$, obtaining $\omega(\mcG_x,\mcS_x) \geq c$.
        \item \textbf{Soundness.} For any $x \in L_{no}$ we have $\omega^*(\mcG_x) <s$.
    \end{enumerate}
We use $\OTS$ to denote the class of all languages $L$ which admits an OTS proof system with a constant completeness-soundness gap.
\end{defn}

Any OTS proof system is described as a special instance of a 2-player, 1-round  ${\MIP}^*$ proof system with additional constraints regarding the completeness condition and we say that an OTS proof system is statistically zero-knowledge if it is statistically zero-knowledge as an ${\MIP}^*$ proof system (see \Cref{defn:SZKMIP*} for details).

\subsection{Applications to ZK and delegated computation}\label{subsec:Applications}
In this section, we show that any language in $\QMA$ admits a statistical zero-knowledge OTS proof system. We also consider how the OTS model can be scaled down to provide a protocol for verifiable delegated quantum computation.

\begin{thm}\label{thm:QMA in Skep}
    For every language $L$ in $\QMA$, there exists a statistical zero-knowledge OTS proof system for $L$ with constant completeness and soundness gap.
\end{thm}

\begin{proof}
We will be working with the proof system sending an instance $x$ to game~$\widehat{\mcG}_x$, as described in~\Cref{xtoGhat}. Using the rigidity results in \Cref{Sec:LWPBT} and \Cref{sec:rigidity}, we have already shown completeness and soundness of properties of the individual game $\widehat{\mcG}_x$ in \Cref{sec:Game}. The ZK property of this game has also been shown in \Cref{Sec:ZK}. All that remains to show is that this protocol further satisfies the extra restrictions of completeness using off-the-shelf devices outlined in \Cref{defn:SQVM-(Formal)}. That is, we need to show that there exists an efficient family of off-the-shelf devices $\lbrace D_n \rbrace$ which can implement the honest strategy $\mcS_x$ for all yes instances $x$.

For each $L \in \QMA$, there exists a polynomial $f$ such that, for all $x \in L$ of size $|x|=n$ the corresponding Hamiltonian~$H_x$ is supported on at most $f(n)$ qubits. Next suppose $x \in L_{yes}$ with $|x| \leq n$. In the honest strategy for the game $\mcG_x$, Alice and Bob share at most $f(n)$-EPR pairs, additionally, Bob privately holds a ground state $\rho$ for~$H_x$. The measurements required by Alice always correspond to $\sigma_X$ or $\sigma_Z$ on up to $6$ qubits of the shared EPR pairs, or $\sigma_X\sigma_Z \otimes \sigma_Z \sigma_X$ on two qubits. In the honest strategy~$\mathcal{S}_x$ for the $m$-fold parallel repeated anchoring game $\widehat{\mcG}_x$, the players share $m f(n)$-EPR pairs and Alice's measures in $\sigma_X$ or $\sigma_Z$ on up to $6m$ qubits or measures with $\sigma_X\sigma_Z \otimes \sigma_Z \sigma_X$ on $2m$ qubits. Since $m=poly(n)$, we then satisfy the completeness condition required by specifying an efficient family of off-the-shelf devices $\lbrace D_n \rbrace$, where for each~$n$ the verification device $D_n$ contains $m f(n)$-EPR pairs and all of the above required Pauli measurements on up to $6m$ qubits.
\end{proof}

In our application towards delegated quantum computation, we consider a novel type of interactive protocol, where in addition to exchanging classical messages, the server can send an untrusted verification device, as defined in \Cref{subsec:FormalModel}, to the client (see~\Cref{sec:intro}). In \Cref{fig:OTSProtocol}, we consider the case where in the first ``message", called a set-up stage, the prover sends an untrusted verification device, which is followed by classical communication.

\begin{figure}
    \centering
    \begin{mdframed}
    \begin{enumerate}
    \item \textbf{Set-up:} The client sends a set-up parameter $k \in \mathbb{N}$ to the server who provides a verification device $D_k$ from an efficient family of off-the-shelf devices $\lbrace D_n \rbrace_n$.
    \item \textbf{Choice of computation:} The client sends a classical description of circuit $Q$, satisfying $|Q|\leq k$ to the server.
    \item \textbf{Verifiable delegation:} The client plays a 1-round game $\widehat{\mcG}_Q$, using the server and device~$D_k$ as players. The client accepts if and only if the game is won.
\end{enumerate}
\end{mdframed}
    \caption{A delegation protocol between a polynomial-time classical client and polynomial-time quantum server, who provides an untrusted verification device during set-up.}
    \label{fig:OTSProtocol}
\end{figure}

Putting these observations together with \Cref{thm:QMA in Skep}, we obtain the following result on delegated quantum computation.

\begin{thm}\label{thm:BQP Delegation}
For every language $L$ in $\BQP$, there is a statistical-zero-knowledge delegation protocol as outlined in \Cref{fig:OTSProtocol} for $L$ with constant completeness and soundness gap.
\end{thm}

\Cref{thm:BQP Delegation} holds largely due to the fact that it can be viewed as a particular application of \Cref{thm:QMA in Skep} in the setting of a $\BQP$-complete problem. The only additional technical requirement occurs in step 3, in which an honest polynomial-time quantum prover is required to be able to saturate the optimal winning probability of the corresponding game. We provide the following sketch.

\begin{proof} \emph{(Sketch)}
We can view the $\BQP$-complete problem from \Cref{defn:Q-circuitProblem} as a language in $\QMA$. This allows us to apply the efficient mapping outlined in \Cref{xtoGhat} to obtain a corresponding game $\widehat{\mcG}_Q$. In this case, the ground state of the underlying Hamiltonian can be prepared by a polynomial-time quantum prover. Thus, as in the proof of \Cref{thm:QMA in Skep}, we can define the required polynomial-time uniform family of off-the-shelf devices $\lbrace D_n \rbrace_{n\in\N}$ by taking $D_n$ to contain suitably many EPR pairs, as well as the required Pauli measurements. Since furthermore the required ground state can always be prepared by a polynomial-time quantum prover, an honest server can obtain the required completeness in Step~3 by generating this state and teleporting it to the verification device when required. We also have that the above delegation protocol inherits the ZK property via the results of \Cref{thm:QMA in Skep}.
\end{proof}

%% file: main.bbl
\newcommand{\etalchar}[1]{$^{#1}$}
\makeatletter\@ifundefined{url}{\newcommand{\url}[1]{\texttt{#1}}}{}\@ifundefined{href}{\newcommand{\href}[2]{\texttt{#2}}}{}\@ifundefined{mathbb}{\newcommand{\mathbb}[1]{#1}}{}\makeatother